\documentclass{article}

\usepackage{amsmath}
\usepackage{amsfonts}
\usepackage{amsthm}
\usepackage{amssymb}
\usepackage{mathrsfs}
\usepackage{mathtools} 
\usepackage{geometry}
\usepackage{graphicx}
\usepackage{hyperref}
\hypersetup{
    colorlinks=true,       
    linkcolor=red,          
    citecolor=blue,        
}

\renewcommand{\Re}{\text{\rm Re}\,}
\renewcommand{\Im}{\text{\rm Im}\,}
\newcommand{\PT}{\mathcal{PT}}
\renewcommand{\P}{\mathcal{P}}
\newcommand{\T}{\mathcal{T}}
\newcommand{\ii}{{\rm i}}

\newcommand{\ee}{{\rm e}}
\newcommand{\hilb}{\mathscr{H}}

\newcommand{\Dom}{\mathrm{Dom}}

\newcommand{\beq}{\begin{equation}}
\newcommand{\eeq}{\end{equation}}
\newcommand{\beqs}{\begin{equation*}}
\newcommand{\eeqs}{\end{equation*}}
\newcommand{\bal}{\begin{aligned}}
\newcommand{\eal}{\end{aligned}}
\newcommand{\dx}{\,\mathrm{d}x}
\newcommand{\dy}{\,\mathrm{d}y}
\newcommand{\du}{\,\mathrm{d}u}
\newcommand{\RR}{\mathbb{R}}
\newcommand{\CC}{\mathbb{C}}
\newcommand{\NN}{\mathbb{N}}

\newcommand{\drho}{\,\textrm{d}\rho}
\newcommand{\JJ}{1}

\newtheoremstyle{proposition} 
    {.7em}                    
    {.5em}                    
    {\itshape}                   
    {}                           
    {\bfseries}                   
    {.}                          
    {\newline}                       
    {}  

\theoremstyle{plain}
\newtheorem{thm}{Theorem}[section] 
\newtheorem{prop}[thm]{Proposition}
\newtheorem{lem}[thm]{Lemma}

\theoremstyle{definition}

\theoremstyle{remark}
\newtheorem{rmrk}[thm]{Remark}

\usepackage[affil-it,auth-sc-lg]{authblk}

\title{Bound states in waveguides with complex Robin boundary conditions}
\date{\today}
\author{Radek Nov\'{a}k\thanks{Electronic address: \texttt{novakra9@fjfi.cvut.cz}}}
\affil{Department of Physics, Faculty of Nuclear Sciences and Physical Engineering, Czech Technical University in Prague, B\v rehov\' a 7, 115 19, Prague, Czech Republic}
\affil{Department of Theoretical Physics, Nuclear Physics Institute, Academy of Sciences of the Czech Republic, Hlavn\' i 130, 250 68 \v Re\v z near Prague, Czech Republic}
\numberwithin{equation}{section}
\providecommand{\keywords}[1]{\textbf{\textit{Key words: }} #1\\[0.5em]}
\providecommand{\classification}[1]{\textbf{\textit{Mathematics Subject Classification (2010): }} #1}
\begin{document}
\maketitle
\begin{abstract}
We consider the Laplacian in a tubular neighbourhood of a hyperplane subjected to non-self-adjoint $\PT$-symmetric Robin boundary conditions. Its spectrum is found to be purely essential and real for constant boundary conditions. The influence of the perturbation in the boundary conditions on the threshold of the essential spectrum is studied using the Birman-Schwinger principle. Our aim is to derive a sufficient condition for existence, uniqueness and reality of discrete eigenvalues. We show that discrete spectrum exists when the perturbation acts in the mean against the unperturbed boundary conditions and we are able to obtain the first term in its asymptotic expansion in the weak coupling regime.
\end{abstract}
\keywords{Non-self-adjointness, waveguide, Robin boundary conditions, spectral analysis, essential spectrum, weak coupling, Birman-Schwinger principle, reality of the spectrum}
\classification{35J05, 35P15, 45C05, 47B44}
%
\section{Introduction} \label{sekcedva}
%
Quantum waveguides undoubtedly belong among the systems interesting both from the physical and mathematical perspective. This notion customarily denotes long and thin semiconductor tubes or layers produced of very pure and crystalline materials. Usually Hamiltonians describing these models are self-adjoint and the bound states correspond to an electron trapped inside the waveguide. One of the possible ways how to describe a transport inside quantum waveguides is to consider the Laplacian in an unbounded tubular region $\Omega$. Physical relevance of such description have been thoroughly discussed in \cite{Duclos-1995, Hurt-2000, Londergan-1999}. The confinement of the wavefuntion to the spatial region is usually achieved by imposing Dirichlet \cite{ExnerSeba, Goldstone}, Neumann \cite{Dittrich-2002, Nazarov} or Robin \cite{ExnerMinakov, Freitas-2006, Jilek07} boundary conditions on $\partial \Omega$. 
\paragraph{}In this paper we choose to study properties of a Laplacian in a tubular neighbourhood of a hyperplane $\RR^n \times I$, where $I=(0,d)$ is a finite one-dimensional interval. Instead of standard self-adjoint boundary condition we impose on the boundary complex Robin boundary conditions
\beq
\frac{\partial \Psi}{\partial n} + \ii \, \alpha \Psi = 0,
\eeq
where $\Psi$ is a wavefunction, $n$ denotes the unit normal vector field of the boundary and $\alpha$ is a real-valued function. The selected boundary conditions physically correspond to the imperfect containment of the electron in the waveguide. This type of boundary conditions has been considered before in the description of open quantum systems \cite{kaiser, raey} and in the context of quantum waveguides in \cite{Borisov08}. (See also \cite{Borisov14-02, Borisov14-01, Borisov12} for other results in this direction.)
\paragraph{}In the paper \cite{Borisov08} the authors focused on the case of the planar waveguide, $n=1$. The spectrum of the waveguide with constant boundary conditions (i.e. $\alpha(x) = \alpha_0$ along the boundary) was found to be purely essential and equal to the half-line $[\mu_0^2, +\infty)$, where $\mu_0^2 := \min\left\{\alpha_0^2,\left(\frac{\pi}{d}\right)^2\right\}$. Furthermore, it is stable under sufficiently smooth compact perturbation $\beta$ of the function $\alpha$. In the case of a weakly coupled perturbation $\varepsilon \beta$ the existence and uniqueness of an isolated eigenvalue was established under the condition that $\alpha_0 \int_{\RR}\beta(x)\dx < 0$ holds and its asymptotic expansion up to the order $\varepsilon^3$ was calculated. The border case $\alpha_0 \int_{\RR}\beta(x)\dx =0$ was studied as well. This paper aims to generalise some of the above mentioned results to higher dimensions and to more general perturbations without compact support. In \cite{Borisov08} method of matched asymptotic expansions was used, we choose a different approach to the problem based on the Birman-Schwinger principle.
\paragraph{}Another reason for choosing complex Robin boundary conditions arises from the context of the so-called $\PT$-symmetric quantum mechanics. Motivated by the numerical observation of purely real spectrum of an imaginary cubic oscillator Hamiltonian \cite{Bender-1998} it blossomed into a large and rapidly developing field studying non-self-adjoint operators. See e.g. \cite{Bender-2007, Mostafazadeh-2010} and reference therein for a survey of papers in this area. The $\PT$-symmetry property of operator $H$ should be here understood as its invariance on the Hilbert space $L^2(\RR^n \times I)$, i.e.
\beq \label{PTsymmetry}
[H,\PT] = 0
\eeq
in the operator sense, where $(\P \Psi)(x,u) := \Psi(x, d-u)$ stands for spatial reflection and $(\T \Psi)(x,u) := \overline{\Psi(x,u)}$ stands for time reversal. The relevant physical interpretation of the operators is ensured when they are in addition quasi-self-adjoint, i.e. they are similar to a self-adjoint operator $h=\omega H \omega^{-1}$, where $\omega$ is a bounded and boundedly invertible operator. Then it is ensured that spectra of $h$ and $H$ are identical and that the corresponding families of eigenfunctions share essential basis properties \cite{Krejcirik12, KrejcirikViola}.
\paragraph{}This paper is organised as follows. In the following section we summarise main results. Section \ref{sekcectyri} is devoted to the proper definition of the Hamiltonian outlined in Section \ref{sekcedva} and to proof of its basic properties. We study essential spectrum of the model in Section \ref{sekcepet}. First of all we study the waveguide with constant boundary conditions along its boundary and their perturbations. Finally, Section \ref{sekcesest} studies the existence of weakly-coupled bound states in this perturbed waveguide.
\section{Main results} \label{sekcetri}
%
Let us consider a region $\Omega:= \RR^n \times I$ embedded into $\RR^{n+1}$, where $I = (0,d)$ is a finite interval. For $n=1$ it reduces to a planar strip, for $n=2$ a layer in three dimensions. We study the problem for a general $n$ except for the investigation of the bound states, where a specific form of the resolvent function of the Hamiltonian plays its role. We are interested in the action of the Hamiltonian of a free particle in this region subjected to $\PT$-symmetric Robin boundary condition on $\partial \Omega$ acting in the Hilbert space $L^2(\Omega)$. Elements of this Hilbert space are going to be consistently denoted with capital Greek letters (usually $\Psi$ or $\Phi$). The variables are going to be split as $(x,u)$, where $x \in \RR^n$ and $u \in (0, d)$. Given a real-valued function $\alpha \in W^{1,\infty}(\RR^n)$ we define the Hamiltonian as
\beq
\bal \label{H-waveguide}
H_{\alpha} \Psi &:= -\Delta \Psi, \\
\Dom(H_{\alpha}) &:= \left\{ \Psi \in W^{2,2}(\Omega) \left|\, \partial_u \Psi + \ii \alpha \Psi = 0 \quad\mathrm{on}\quad \partial\Omega \right.\right\},
\eal
\eeq
where $\partial_u$ stands for differentiation with respect to $u$, similarly $\Delta$ stands for sum of all second derivatives. The effect of $H_{\alpha}$ should be understood in a distributional sense and the boundary conditions in the sense of traces.
\paragraph{}We can see that the probability current in $\RR^{n+1}$ of wavefunction $\Psi \in \Dom(H_{\alpha})$ gives in the point $(x,u)$ of $\partial \Omega$
\beq
\bal
\vec{j}(x,u) = \frac{1}{\ii}\left( \overline{\Psi} \partial_u \Psi - \Psi \partial_u\overline{\Psi}\right)(x,u) \,\vec{e}_{n+1} = -2 \alpha(x) |\Psi(x,u)|^2 \,\vec{e}_{n+1}, 
\eal
\eeq
where $\vec{e}_{n+1}$ stands for $(n+1)$-th vector of the standard basis in $\RR^{n+1}$. Clearly the current is not equal to zero for non-trivial $\alpha$ and general $\Psi$. However, the influence of the boundary conditions on the current does not depend on whether we are at $u = 0$ or $u = d$ and therefore is the same for both components of $\partial \Omega$ and the gain and loss are balanced.
\paragraph{}Using the quadratic form approach and the First representation theorem, it will be derived in Theorem \ref{m-sec-waveguide} that $H_{\alpha}$ is an m-sectorial operator if $\alpha \in W^{1,\infty}(\RR^n)$.
This yields that the operator is closed, therefore its spectrum is well defined and contained in a sector. Furthermore, the spectrum of $H_{\alpha}$ is localised inside a parabola, more precisely, 
\beq
\sigma(H_{\alpha}) \subset \left\{ z \in \CC \left| \Re z \geq 0, |\Im z| \leq 2 \|\alpha\|_{L^{\infty}(\RR^n)} \sqrt{\Re z}\right.\right\}.
\eeq
Using the quadratic forms it can be shown for its adjoint operator that $H_{\alpha}^* = H_{-\alpha}$. Note that $H_{\alpha}$ is not self-adjoint, unless $\alpha$ is identically equal to $0$. 
\paragraph{}Elementary calculations also lead to the conclusion, that $H_{\alpha}$ is $\PT$-symmetric, i.e. commutes with operator $\PT$ in operator sense explained in \cite[Sec. III.5.6]{Kato}. The spatial reflection operator $\P$ and the time reversal operator $\T$ are in our context defined as
\beq
\bal \label{def-pt}
(\P \Psi)(x,u) &:= \Psi(x, d-u),\\
(\T \Psi)(x,u) &:= \overline{\Psi(x,u)}.
\eal
\eeq
\paragraph{}Another important property of $H_{\alpha}$ is $\T$-selfadjointness, i.e $\T H_{\alpha} \T = H_{\alpha}^*$. A major consequence of this is that the residual spectrum of $H_{\alpha}$ is empty \cite[Cor. 2.1]{Borisov08}, i.e
\beq
\sigma_{\mathrm{r}}(H_{\alpha})= \varnothing.
\eeq
We emphasize that in our non-self-adjoint case it was impossible to a priori say anything about the residual spectrum, compared to the self-adjoint case, in which it is always empty.
\paragraph{}Before approaching deeper results, we focus on a very simple case of the boundary conditions, $\alpha(x) = \alpha_0$ for all $x \in \RR^n$, where $\alpha_0$ is a real constant. Using the decomposition of the resolvent into the transversal basis, it is possible to show that the Hamiltonian $H_{\alpha_0}$ can be written as a sum
 \beq \label{decomposition}
 H_{\alpha_0} = -\Delta'\otimes \JJ^{I} + \JJ^{\RR^n} \otimes -\Delta^{I}_{\alpha_0},
 \eeq
 where $\JJ^{\RR^n}$ and $\JJ^{I}$ are identity operators on $L^2(\RR^n)$ and $L^2(I)$ respectively, $-\Delta'$ is a self-adjoint Laplacian in $L^2(\RR^n)$ and $-\Delta^{I}_{\alpha_0}$ is a Laplacian in $L^2(I)$ with complex Robin-type boundary conditions (see \eqref{transversal-operator} for a precise definiciton). The latter operator has been extensively studied in \cite{Hernandez-2011, Hussein, Krejcirik-2008-41a, Krejcirik-2006-39, Krejcirik-2010-43}. It was shown that it is an m-sectorial and quasi-self-adjoint operator. It has purely discrete spectrum, its lowest lying point we denote as $\mu_0^2$. It holds that $\mu_0^2 := \min\left\{\alpha_0^2,\left(\frac{\pi}{d}\right)^2\right\}$. Our main conclusion about the spectrum of $H_{\alpha_0}$ is the following:
\begin{prop} \label{noper}
Let $\alpha_0 \in \RR$. Then
\beq
\sigma(H_{\alpha_0}) = \sigma_{\mathrm{ess}}(H_{\alpha_0}) = [\mu_0^2,+\infty).
\eeq
\end{prop}
\begin{rmrk}
There are several different definitions of the essential spectra in literature. For the self-adjoint operators they coincide, however this needs not to be true when the operator is non-self-adjoint and the various essential spectra can differ significantly. We employ the definition via so-called singular sequences - For a closed operator $A$ we say that $\lambda \in \mathbb{C}$ belongs to the essential spectrum of $A$ (denoted $\sigma_{\mathrm{ess}}(T)$) if there exists a sequence $(\psi_n)_{n=1}^{+\infty}$ (called a singular sequence), $\|\psi_n\|_{\hilb} = 1$ for all $n$, such that it does not contain any convergent subsequence and $\lim_{n\rightarrow +\infty}(T-\lambda)\psi_n = 0$. Other definitions are based e.g. on the violation of the Fredholm property (i.e. range of the studied operator is not closed or its kernel or cokernel are not finite-dimensional). However, many of these definitions coincide, provided $A$ is $\T$-self-adjoint\cite[Thm.IX.1.6]{EE}.
\end{rmrk}
Further on we study the perturbed waveguide, where the function $\alpha$ from the boundary conditions takes the form
\beq \label{decom}
\alpha(x) = \alpha_0 + \varepsilon \beta(x).
\eeq
Here $\beta \in W^{2,\infty}(\RR^n)$ and $\varepsilon >0$. The stability of the essential spectrum is ensured when the boundary conditions approach uniform boundary conditions in infinity.
\begin{thm} \label{per}
Let $\alpha - \alpha_0 \in W^{1,\infty}(\RR)$ with $\alpha_0 \in \RR$ such that
\beq \label{constantinfinity}
\lim_{|x| \rightarrow +\infty} (\alpha-\alpha_0)(x) = 0
\eeq
Then
\beq
\sigma_{\mathrm{ess}}(H_{\alpha}) = \sigma_{\mathrm{ess}}(H_{\alpha_0}) = [\mu_0^2,+\infty). 
\eeq
\end{thm}
In the rest of the paper we search for conditions under which a small perturbation allows the existence of a bound state, i.e. of an isolated eigenvalue with finite geometric multiplicity. Due to the singularity of the resolvent this effect can be expected when the effective infinite dimension of the problem is $1$ or $2$. (See Remark \ref{singularita-rmrk} for more details.) Our method of ensuring its existence works under assumtion of a sufficiently fast decay of $\beta$ in infinity, which is summarized in technical conditions (\ref{predpoklady}) and (\ref{predpoklady2}). Using different estimates in the proofs of relevant lemmas it could be probably improved. In further text the mean value of $\beta$ is denoted as $\langle\beta\rangle := \int_{\RR^n} \beta(x) \dx$. 
\begin{thm} \label{vetavet}
Let us recall \eqref{decom} and assume \eqref{predpoklady} if $n=1$ or \eqref{predpoklady2} if $n=2$ with $\beta \in W^{2,\infty}(\RR^n)$. If $\varepsilon > 0$ is sufficiently small, $|\alpha_0| < \pi/d$, then $H_{\alpha}$ possesses a unique, simple and real eigenvalue $\lambda = \lambda(\varepsilon) \in \mathbb{C}\setminus [0,+\infty)$ if $\alpha_0\langle\beta\rangle < 0$. The asymptotic expansion
\begin{equation} \label{expansion}
\lambda(\varepsilon) = 
\begin{dcases}
\mu_0^2 - \varepsilon^2 \alpha_0^2 \langle\beta\rangle^2 + \mathcal{O}(\varepsilon^3), \\
\mu_0^2 - \ee^{2/w(\varepsilon)},
\end{dcases}
\end{equation}
where $w(\varepsilon) = \frac{\varepsilon}{\pi}  \alpha_0\langle\beta\rangle + \mathcal{O}(\varepsilon^2)$, holds as $\varepsilon \rightarrow 0$. If $\alpha_0\langle\beta\rangle > 0$, $H_{\alpha}$ has no eigenvalues.
\end{thm}
When $\alpha_0 > \pi/d$, \eqref{neco} is equal to zero and we are unable to say anything about the eigenvalue. To do so it would be necessary to take higher terms in the expansion of $\lambda$, which shows to be computationally challenging by the present method. We would encounter similar difficulties when trying to obtain more than just the leading term in the asymptotic expansion\eqref{expansion} to check the equality situation $ \alpha_0 \langle \beta\rangle = 0$.
\paragraph{}We have just seen that the existence of the weakly coupled bound state is conditioned by fulfilment of $\alpha_0\langle\beta\rangle < 0$. Both $\alpha_0$ and $\beta$ play equivalent role in the boundary conditions - they cause a non-zero probability current over each component of the boundary. However, the negative sign of their product means, that they generate the probability current against each other. We may conclude that the weakening of the probability current through the waveguide due to the small perturbation is responsible for the existence of the bound state. 
%
\section{Definition of the Hamiltonian} \label{sekcectyri}
%
This section is devoted to a proper definition of the Hamiltonian outlined in Sections \ref{sekcedva} and \ref{sekcetri} and to stating its basic properties. We begin by prescription of the densely defined sesquilinear form
\beq
\bal
h_{\alpha}(\Phi,\Psi) &:= h^1_{\alpha}(\Phi,\Psi) + \ii\, h^2_{\alpha}(\Phi,\Psi)\\
\Dom(h_{\alpha}) &:= W^{1,2}(\Omega)
\eal
\eeq
where the real part $h^1_{\alpha}$ and the imaginary part $h^2_{\alpha}$ are two sesquilinear forms defined on $W^{1,2}(\Omega)$ as
\beq
\bal
h^1_{\alpha}(\Phi,\Psi) &:= \int_{\Omega} \overline{\nabla \Phi(x,u)} \cdot \nabla\Psi(x,u) \dx \du,\\
h^2_{\alpha}(\Phi,\Psi) &:= \int_{\RR^n} \alpha(x)\,\overline{\Phi(x,d)} \,\Psi(x,d)\dx - \int_{\RR^n} \alpha(x)\,\overline{\Phi(x,0)} \,\Psi(x,0)\dx,
\eal
\eeq
where the dot stands for the scalar product in $\RR^n$ and the boundary term should be again understood in the sense of traces. The form $h^1_{\alpha}$ is associated with a Neumann Laplacian in $L^2(\Omega)$, it is therefore densely defined, closed, positive and symmetric. In the spirit of perturbation theory we show that $h^2_{\alpha}$ plays a role of a small perturbation of $h^1_{\alpha}$. We employ the notation $h[\cdot]$ for the quadratic form associated with the sesquilinear form $h(\cdot,\cdot)$.
\begin{lem} \label{prvni-lemma}
Let $\alpha \in L^{\infty}(\RR^n)$. The $h^2_{\alpha}$ is relatively bounded with respect to $h^1_{\alpha}$ with arbitrarily small relative bound. We have
\beq
\left|h^2_{\alpha}[\Psi]\right| \leq 2 \|\alpha\|_{L^{\infty}(\RR^n)}\|\Psi\|_{L^2(\Omega)} \sqrt{h^1_{\alpha}[\Psi]} \leq \delta h_{\alpha}^1[\Psi] + \frac{1}{\delta} \|\alpha\|_{L^{\infty}(\RR^n)}^2 \|\Psi\|_{L^2(\Omega)}^2
\eeq
for every $\Psi \in W^{1,2}(\Omega)$ and $\delta >0$.
\end{lem}
\begin{proof}
Since $\Omega$ satisfies the segment condition, the set of restrictions of $C^{\infty}_0(\RR^n)$ functions to $\Omega$ is dense in $W^{1,2}(\Omega)$ \cite[Thm. 3.22]{Adams1975}. (To check the condition, it is sufficient to take as $U_x$ a ball with radius strictly smaller than $d/2$ and as the vector $y_x$ any inwards pointing vector not exceeding the length of $d/2$.) We may thus restrict ourselves to the case $\Psi\in C^{\infty}_0(\RR^n)$. Now we are able to differentiate $|\Psi(x)|^2$ and hence we may write
\beq
\bal
\left|h^2_{\alpha}[\Psi]\right| &= \left|\int_{\Omega} \alpha(x)\frac{\partial |\Psi(x,u)|^2}{\partial u}\dx\du\right| \\
&\leq 2 \|\alpha\|_{L^{\infty}(\RR^n)} \int_{\Omega} |\Psi(x,u)||\partial_u \Psi(x,u)|\dx\du \\
&\leq 2 \|\alpha\|_{L^{\infty}(\RR^n)} \|\Psi\|_{L^2(\Omega)} \|\partial_u \Psi\|_{L^2(\Omega)}\\
&\leq 2 \|\alpha\|_{L^{\infty}(\RR^n)}\|\Psi\|_{L^2(\Omega)} \sqrt{h^1_{\alpha}[\Psi]},
\eal
\eeq
where we used the inequality $\|\partial_u \Psi\|_{L^2(\Omega)} \leq \|\nabla \Psi\|_{L^2(\Omega)} = \sqrt{h^1_{\alpha}[\Psi]}$. On this result we apply the Young inequality and we obtain the other inequality from the claim.
\end{proof}
According to \cite[Thm. VI-1.33]{Kato}, the form $h_{\alpha}$ is closed and sectorial. The First Representation Theorem \cite[Thm. VI-2.1]{Kato} states that then there exists a unique m-sectorial operator $\tilde H_{\alpha}$ such that $h_{\alpha}(\Phi,\Psi) = (\Phi, \tilde H_{\alpha} \Psi)_{L^2(\Omega)}$ for all $\Psi \in \Dom(\tilde H_{\alpha}) \subset \Dom(h_{\alpha})$ and $\Phi \in \Dom(h_{\alpha})$. The domain of $\tilde H_{\alpha}$ can be expressed as
\beq
\Dom(\tilde H_{\alpha}) = \left\{ \Psi \in W^{1,2}(\Omega)\left| \exists F \in L^2(\Omega), \forall \Phi \in W^{1,2}(\Omega), h_{\alpha}(\Phi,\Psi) = (\Phi, F)_{L^2(\Omega)}\right.\right\}
\eeq
To prove that $\tilde H_{\alpha} = H_{\alpha}$, we state first an auxiliary lemma.
\begin{lem} \label{m-sec-lemma}
Let $\alpha \in W^{1,\infty}(\RR^n)$. For each $F \in L^2(\Omega)$ a solution $\Psi$ to the problem
\beq \label{lemma-rovnice}
h_{\alpha}(\Phi,\Psi) = (\Phi,F)_{L^2(\Omega)}
\eeq
for all $\Phi \in W^{1,2}(\Omega)$ belongs to $\Dom(H_{\alpha})$.
\end{lem}
\begin{rmrk}
Equivalently, the statement may be formulated that the generalized solution to the problem
\beq
\begin{dcases}
\qquad-\Delta \Psi \!\!\!&= F \qquad \mathrm{in} \; \Omega \\
\partial_u \Psi + \ii \alpha \Psi\!\!\! &= 0 \qquad \mathrm{on} \; \partial\Omega
\end{dcases}
\eeq 
belongs to $\Dom(H_{\alpha})$.
\end{rmrk}
\begin{proof}
We introduce the difference quotient \cite[Sec. 5.8.2]{Evans-1998}
\beq
\Psi^j_{\delta}(x,u) := \frac{\Psi(x + \delta e_j, u) - \Psi(x, u)}{\delta}
\eeq
for $j=1,\dots,n$ and any $\Psi \in L^2(\Omega)$ and $\delta$ a small real number. Here $e_j$ stands for $j$-th vector of the standard basis in $\RR^n$, i.e. $x + \delta e_j = (x_1, \dots, x_{j-1}, x_j + \delta, x_{j+1}, \dots, x_n)$. We estimate using the Schwarz inequality
\beq
\bal
\left|\Psi(x+\delta e_j,u)-\Psi(x,u)\right| = \left|\delta \int_0^1 \partial_{x_j} \Psi(x+ \delta e_j,u) \,\mathrm{d}t \right|
&\leq |\delta| \sqrt{\int_0^1 |\partial_{x_j} \Psi(x + \delta e_j t)|^2\,\mathrm{d}t}, 
\eal
\eeq
which subsequently with the use of Fubini's theorem yields the inequality
\beq
\bal
 \|\Psi^j_{\delta}\|_{L^2(\Omega)}^2 \leq \int_{\Omega} \left( \int_0^1 |\partial_{x_j} \Psi(x + \delta e_j t)|^2\,\mathrm{d}t\right) \dx \du = \int_{0}^1  \|\partial_{x_j} \Psi\|_{L^2(\Omega)} \mathrm{d}t \leq \|\Psi\|_{W^{1,2}(\Omega)}^2.
\eal
\eeq
Similarly we estimate $\alpha_{\delta}^j$:
\beq
\bal
\|\alpha_{\delta}^j\|_{L^{\infty}(\RR^n)} \leq \underset{x \in \RR^n}{\mathrm{ess\,sup}} \int_0^1 |\partial_{x_j} \alpha(x + \delta e_j t)| \,\mathrm{d}t \leq \|\partial_{x_j}\alpha\|_{L^{\infty}(\RR^n)} \leq \|\alpha\|_{W^{1,\infty}(\RR^n)}.
\eal
\eeq
If $\Psi$ satisfies \eqref{lemma-rovnice}, then $\Psi_{\delta}$ is a solution to 
\beq \label{cosi}
h_{\alpha}(\Phi,\Psi^j_{\delta}) = \left(\Phi, F^j_{\delta}\right)_{L^2(\Omega)} - \int_{\RR^n} \alpha^j_{\delta}(x) \left( \overline{\Phi(x,0)}\Psi(x+\delta e_j,0) - \overline{\Phi(x,d)}\Psi(x+\delta e_j,d)\right) \dx
\eeq
with $\Phi \in W^{1,2}(\Omega)$ arbitrary. It also holds
\beq
\bal
\left(\Phi, F^j_{\delta}\right)_{L^2(\Omega)}&= \frac{1}{\delta}\int_{\Omega} \overline{\Phi(x,u)} \left(F(x+\delta e_j,u) - F(x,u) \right)\dx \du \\
&= \frac{1}{\delta}\int_{\Omega} \left( \overline{\Phi(x-\delta e_j,u)} - \overline{\Phi(x,u)}\right) F(x,u) \dx \du \\
& = -\left(\Phi^j_{-\delta},F\right)_{L^2(\Omega)}
\eal
\eeq
and we use it together with setting $\Phi = \Psi^j_{\delta}$ to obtain from \eqref{cosi}
\beq
\bal
h_{\alpha}[\Psi^j_{\delta}] =& - \left( (\Psi^j_{\delta})_{-\delta},F\right)_{L^2(\Omega)} \\
&- \int_{\RR^n} \alpha^j_{\delta}(x) \left( \overline{\Psi^j_{\Delta}(x,0)}\Psi(x+\delta e_j,0) - \overline{\Psi^j_{\Delta}(x,d)}\Psi(x+\delta e_j,d)\right) \dx.
\eal
\eeq
We employ the estimates
\beq
\left|\left( (\Psi^j_{\delta})^j_{-\delta},F\right)\right| \leq \|F\|_{L^2(\Omega)} \|(\Psi^j_{\delta})^j_{-\delta}\|_{L^(\Omega)} \leq \frac{1}{2} \|F\|_{L^2(\Omega)} + \frac{1}{2}\|\Psi^j_{\delta}\|_{W^{1,2}(\Omega)}
\eeq
and
\beq
\bal
\left| \int_{\RR^n} \right.&\left. \alpha^j_{\delta}(x) \left( \overline{\Psi^j_{\delta}(x,0)}\Psi(x+\delta e_j,0) - \overline{\Psi^j_{\delta}(x,d)}\Psi(x+\delta e_j,d)\right) \dx \right| \\
 &\leq 2 \|\alpha\|_{W^{1,\infty}(\RR^n)}\|T \Psi^j_{\delta}\|_{L^2(\partial \Omega)} \|T \Psi\|_{L^2(\partial \Omega)} \\
 &\leq C_1 \|\Psi^j_{\delta}\|_{W^{1,2}(\Omega)} \|\Psi\|_{W^{1,2}(\Omega)}
\eal
\eeq
where $T$ is trace operator $W^{1,2}(\Omega) \rightarrow L^2(\partial \Omega)$, together with Young inequality and Lemma \ref{prvni-lemma} to obtain
\beq
\bal
\|\Psi_{\delta}^j\|_{W^{1,2}(\Omega)}^2 =& \|\Psi_{\delta}\|^2_{L^2(\Omega)} + \|\nabla \Psi_{\delta}\|^2_{L^2(\Omega)} \\
 \leq& \|\Psi\|^2_{W^{1,2}(\Omega)} + \frac{1}{2} \|F\|_{L^2(\Omega)} + \frac{1}{2}\|\Psi^j_{\delta}\|_{W^{1,2}(\Omega)} + C_1\|\Psi^j_{\delta}\|_{W^{1,2}(\Omega)} \|\Psi\|_{W^{1,2}(\Omega)} \\
&+ 2 \|\alpha\|_{W^{1,\infty}(\RR^n)} \|\Psi_{\delta}^j\|_{L^2(\Omega)}\|\Psi_{\delta}^j\|_{W^{1,2}(\Omega)}\\
\leq&  \|\Psi\|^2_{W^{1,2}(\Omega)} + \frac{1}{2} \|F\|_{L^2(\Omega)} + \frac{1}{2}\|\Psi^j_{\delta}\|_{W^{1,2}(\Omega)}\\
&+ C_1 \left( \frac{1}{4\tau}\|\Psi\|^2_{W^{1,2}(\Omega)}  + \tau \|\Psi^j_{\delta}\|_{W^{1,2}(\Omega)}\right)\\
&+ C_2 \left( \frac{1}{4\tau}\|\Psi\|^2_{W^{1,2}(\Omega)}  + \tau \|\Psi^j_{\delta}\|_{W^{1,2}(\Omega)}\right)\\
\leq&\frac{1}{2} \|F\|_{L^2(\Omega)} + \left(1 + \frac{C_1 + C_2}{4 \tau}\right)\|\Psi\|^2_{W^{1,2}(\Omega)} + \left(\frac{1}{2}+ (C_1 + C_2)\tau\right)\|\Psi^j_{\delta}\|_{W^{1,2}(\Omega)},
\eal
\eeq
where $\tau >0$ can be chosen arbitrarily small. Setting $\tau = 1/(4 C_1+4 C_2)$ we have
\beq
\|\Psi_{\delta}^j\|_{W^{1,2}(\Omega)} \leq C,
\eeq
where $C$ is independent of $\delta$. This implies that 
\beq
\sup_{\delta\in\RR} \|\Psi_{\delta}\|_{W^{1,2}(\Omega)} < + \infty.
\eeq
Since bounded sequences in a reflexive Banach space are weakly precompact \cite[Thm. D.4.3]{Evans-1998}, we find a subsequence $(\delta_k)_{k=1}^{\infty}$, $\lim_{k\rightarrow +\infty} \delta_k = 0$, such that $\Psi^j_{\delta_k}$ weakly converges to some $f$ in $W^{1,2}(\Omega)$. As can be expected,
\beq
\bal
-\int_{\Omega}\overline{\partial_{x_j}\Psi(x,u)} \Phi(x,u) &=\int_{\Omega} \overline{\Psi(x,u)} \lim_{\delta_k \rightarrow 0} \Phi^j_{-\delta_k}(x,u) \dx \du \\
&= \lim_{\delta_n \rightarrow 0} \int_{\Omega} \overline{\Psi(x,u)} \Phi^j_{-\delta_k}(x,u) \dx \du \\
&= - \lim_{\delta_k \rightarrow 0} \int_{\Omega} \overline{\Psi^j_{\delta_k}(x,u)} \Phi(x,u) \dx \du \\
&= - \int_{\Omega} \overline{f(x,u)} \Phi(x,u) \dx \du.
\eal
\eeq
Therefore $\partial_{x_j}\Psi = f$ in a weak sense and so $\partial_{x_j}\Psi \in W^{1,2}(\Omega)$ for every $j$, $j = 1, \dots, n$. From the Interior Regularity Theorem \cite[Thm. 6.3.1]{Evans-1998} follows that $\Psi \in W^{2,2}_{\mathrm{loc}}(\Omega)$. Hence, the equation $-\Delta \Psi = F$ holds almost everywhere in $\Omega$. Also, $\partial_u^2 \Psi = -F -\Delta' \Psi \in L^2(\Omega)$ and therefore $\Psi \in W^{2,2}(\Omega)$.\\
Using Gauss-Green theorem we find that
\beq
\bal
(\Phi,F)_{L^2(\Omega)} =& (\Phi,-\Delta \Psi)_{L^2(\Omega)} \\
&+\int_{\RR^n} \left( \partial_u \Psi(x,d) + \ii\, \alpha(x)\Psi(x,d)\right)\overline{\Phi(x,d)}\dx\\
&-\int_{\RR^n} \left( \partial_u \Psi(x,0) + \ii\, \alpha(x)\Psi(x,0)\right)\overline{\Phi(x,0)}\dx
\eal
\eeq
for all $\Phi \in W^{1,2}(\Omega)$. Using this equality and the fact that $F = -\Delta \Psi$ almost everywhere in $\Omega$ we obtain the boundary conditions for $\Psi$. 
\end{proof}
\begin{thm} \label{m-sec-waveguide}
Let $\alpha \in W^{1,\infty}(\RR^n)$ be real-valued. Then $H_{\alpha}$ is an m-sectorial operator on $L^2(\Omega)$ satisfying
\beq
H_{\alpha} = \tilde H_{\alpha}.
\eeq
\end{thm}
\begin{proof}
Using integration by parts it is straightforward to verify that $\tilde H_{\alpha}$ is an extension of $H_{\alpha}$, $H_{\alpha} \subset \tilde H_{\alpha}$. The other inclusion follows from Lemma \ref{m-sec-lemma} and the uniqueness in the First Representation Theorem \cite[Thm. VI-2.1]{Kato}.
\end{proof}
Using the quadratic form approach, we are able to find the adjoint operator to $H$ quite easily.
\begin{thm}
Let $\alpha \in W^{1,\infty}(\RR^n)$ be real-valued. Then
\beq \label{wave-adjoint}
H_{\alpha}^* = H_{-\alpha}.
\eeq
\end{thm}
\begin{proof}
We find the adjoint operator $H_{\alpha}^*$ as an operator corresponding to the adjoint form $h_{\alpha}^*$. The adjoint form can be obtained from $h_{\alpha}$ by replacing $\alpha$ for $-\alpha$. Therefore, its corresponding operator is $H_{-\alpha}$.
\end{proof}
Spectrum of $H_{\alpha}$ is indeed well defined since $H_{\alpha}$ is a closed operator. Consequence of $H_{\alpha}$ being m-sectorial is enclosure of its spectrum in a sector in a complex plane. Using the estimate from Lemma \ref{prvni-lemma}, this estimate can be further improved as follows.
\begin{prop}
The spectrum of $H_{\alpha}$ is localised inside a parabola, more precisely, 
\beq
\sigma(H_{\alpha}) \subset \left\{ z \in \CC \left| \Re z \geq 0, |\Im z| \leq 2 \|\alpha\|_{L^{\infty}(\RR^n)} \sqrt{\Re z}\right.\right\}.
\eeq
\end{prop}
The studied Hamiltonian is fundamentally non-self-adjoint, we can however state some symmetry properties, more precisely the $\PT$-symmetry and $\T$-self-adjointness.
\begin{prop} \label{PTH}
Let $\alpha \in W^{1,\infty}(\RR^n)$ be real-valued. Then $H_{\alpha}$ is $\PT$-symmetric with operators $\P$, $\T$ defined in \eqref{def-pt}.
\end{prop}
\begin{proof}
According to our definition \eqref{PTsymmetry} of $\PT$-symmetry we need to check that $[H_{\alpha},\PT] = 0$ holds in the sense $\PT H_{\alpha} \subset H_{\alpha}\PT$ \cite[Sec. III.5.6]{Kato}. For every $\Psi\in \Dom(H_{\alpha})$ easily holds that $\PT \Psi \in W^{2,2}(\Omega)$. We can directly check that the action of $H_{\alpha}$ is invariant under the influence of the operator $\PT$ and that the boundary conditions hold also for $\PT \Psi$.
\end{proof}
\begin{prop}
Let $\alpha \in W^{1,\infty}(\RR^n)$ be real-valued. Then $H_{\alpha}$ is $\T$-self-adjoint, i.e.
\beq
\T H_{\alpha} \T = H_{\alpha}^*
\eeq
\end{prop}
\begin{proof}
The proof follows in the same way as the proof of Proposition \ref{PTH}.
\end{proof}
The $\T$-self-adjointness in particular due to \cite[Cor. 2.1]{Borisov08} implies that 
\beq
\sigma_{\mathrm{r}}(H_{\alpha}) = \varnothing.
\eeq
%
\section{The essential spectrum} \label{sekcepet}
%
\subsection{Uniform boundary conditions} \label{unperturbed-waveguide}
%
Let us now study the operator $H_{\alpha}$ with $\alpha(x)$ identically equal to $\alpha_0\in\RR$ for all $x \in \RR^n$. We are going to establish some of its basic properties and use them in next subsection to study the perturbed operator $H_{\alpha_0 + \varepsilon \beta}$. Our first goal is to prove the decomposition \eqref{decomposition}. Let us summarise some properties of the operator 
\beq
 \bal \label{transversal-operator}
 -\Delta^I_{\alpha_0}\psi &:= -\psi'' \\
 \Dom(-\Delta^I_{\alpha_0}) &:= \left\{ \psi \in W^{2,2}(I)\left|\, \psi' + \ii \alpha_0 \psi = 0 \quad \textrm{at} \quad \partial I\right.\right\}.
 \eal
 \eeq
It has been shown in \cite[Prop. 1]{Krejcirik-2006-39} that it is an m-sectorial operator therefore it is also closed and the study of its spectrum has a good meaning. The point spectrum of $-\Delta^I_{\alpha_0}$ is the countable set $\left\{\mu_j^2 \right\}_{j=0}^{+\infty}$ with
\beq
\bal
\mu_{j_0} := \alpha_0, \qquad \qquad \qquad
\mu_{j_1} := \frac{\pi}{d}, \qquad \qquad \qquad
\mu_j := \frac{j \pi}{d},
\eal
\eeq
where $j \geq 2$, $(j_0,j_1) = (0,1)$ if $|\alpha_0| \leq \pi/d$ and $(j_0,j_1) = (1,0)$ if $|\alpha_0| > \pi/d$. Making the hypothesis 
\beq  \label{simple-spectrum}
\frac{\alpha_0 d}{\pi} \notin \mathbb{Z} \setminus \{0\}
\eeq
the eigenvalues have algebraic multiplicity equal to one. The corresponding set of eigenfunctions $\{\psi_j\}_{j=0}^{+\infty}$ can be chosen as
\beq \label{function}
\psi_j(u) := \cos(\mu_j u) - \ii \frac{\alpha_0}{\mu_j}\sin(\mu_j u), \qquad j \geq 0.
\eeq
Since the resolvent of the operator $-\Delta^I_{\alpha_0}$ is compact \cite[Prop. 2]{Krejcirik-2006-39}, the spectrum is purely discrete and we have
\beq
\sigma(-\Delta^I_{\alpha_0}) = \sigma_{\textrm{d}}(-\Delta^I_{\alpha_0})= \{\mu_j^2\}_{j=0}^{+\infty}.
\eeq
The adjoint operator $(-\Delta^I_{\alpha_0})^*$ possesses the same spectrum since it can be obtained by interchanging $\alpha_0$ for $-\alpha_0$ in the boundary conditions because $-\Delta^I_{\alpha_0}$ fulfils the relations analogous to the one in the equation (\ref{wave-adjoint}), $(-\Delta^I_{\alpha_0})^* = -\Delta^I_{-\alpha_0}$, and therefore the eigenvalue equation remains unchanged. The corresponding eigenfunction can be selected as
\beq \label{adjoint-function}
\phi_j(u) := \overline{A_j \psi_j(u)},
\eeq
where $A_j$ are normalisation constants defined as
\beq
\bal
A_{j_0} := \frac{2 \ii \alpha_0}{1 - \exp(-2\ii \alpha_0 d)},\qquad 
A_{j_1} := \frac{2 \mu_{j_1}^2}{(\mu_{j_1}^2 - \alpha_0^2)d},\qquad
A_{j} := \frac{2 \mu_j^2}{(\mu_j^2 - \alpha_0^2)d},
\eal
\eeq
where $j\geq 2$, $(j_0,j_1) = (0,1)$ if $|\alpha_0| < \pi/d$ and $(j_0,j_1) = (1,0)$ if $|\alpha_0| > \pi/d$. (Note that we already ruled out the case $|\alpha_0| = \pi/d$ due to (\ref{simple-spectrum}).) If $\alpha_0 = 0$, $A_{j_0}$ should be understood in the limit sense $\alpha_0 \rightarrow 0$. With this choice of normalization constants the both sets of eigenvectors form biorthonormal basis \cite[Prop. 3]{Krejcirik-2006-39} with the relations 
\beq \label{biortho-relation}
(\phi_j, \psi_k)_{L^2(I)} = \delta_{jk}, \qquad \forall j,k \in \mathbb{N},
\eeq
and 
\beq \label{transversal-decomposition}
\psi = \sum_{j=0}^{+\infty}(\phi_j, \psi)_{L^2(I)} \psi_j
\eeq
for every $\psi \in L^2(I)$.
\begin{prop}
The identity
\beq \label{bi-expansion}
\Psi(x,u) = \sum_{j=0}^{+\infty} \Psi_j(x) \psi_j(u).
\eeq
where $\Psi_j(x) := \left(\phi_j,\Psi(x,\cdot)\right)_{L^2(I)}$ holds for every $\Psi \in L^2(\Omega)$ in the sense of $L^2(\Omega)$-norm.
\end{prop}
\begin{proof}
The proof follows in the same as in \cite[Lem. 4.1]{Borisov08}. All we need to do is use the dominated convergence theorem, since we already know that the sum \eqref{bi-expansion} converges to $\Psi$ for almost every $x \in \RR^n$ thanks to \eqref{transversal-decomposition}. We estimate the partial sum of \eqref{bi-expansion}. Let us introduce
\begin{equation}
\chi_j^D(u):= \sqrt{\frac{2}{d}}\sin(\frac{\pi j}{d}u) \quad\mbox{if } j \geq 1
\end{equation}
and
\begin{equation}
\chi_j^N(u) :=
\begin{dcases}
\frac{1}{\sqrt{d}} & \mbox{if } j = 0,\\
\sqrt{\frac{2}{d}}\cos(\frac{\pi j}{d} u)  & \mbox{if } j \geq 1.
\end{dcases}
\end{equation}
Recall that $(\chi_j^D)_{j=1}^{+\infty}$ and $(\chi_j^N)_{j=0}^{+\infty}$ form complete orthonormal systems of functions in $L^2(I)$. The eigenfunctions \eqref{function} of the transversal operator can be expressed by the means of $\chi_j^D$ and $\chi_j^N$ as
\beq
\psi_j(u) = \sqrt{\frac{d}{2}}\left(\chi_j^N(u) - \ii \frac{\alpha_0}{\mu_j} \chi_j^D(u)\right)
\eeq
for $j \geq 2$. Using the parallelogram identity and the orthogonality we obtain
\beq
\bal
\left\|\sum_{j=2}^k \Psi_j(x)\psi_j\right\|^2_{L^2(I)}&\leq d \left\|\sum_{j=2}^k \Psi_j(x)\chi_j^N \right\|^2_{L^2(I)} + d \alpha_0^2\left\|\sum_{j=2}^k \Psi_j(x) \chi_j^D \right\|^2_{L^2(I)} \\
&\leq d \left(1 + \frac{\alpha_0^2}{\mu_2^2}\right) \sum_{j=2}^{+\infty} |\Psi_j(x)|^2.
\eal
\eeq
Further on, we can write $\Psi_j(x)$ as $\Psi_j(x)  = \sqrt{\frac{d}{2}} A_j\left( \Psi_j^N - \ii \frac{\alpha_0}{\mu_j}\Psi_j^D\right)$ with $\Psi_j^{\sharp}(x) := (\chi_j^{\sharp},\Psi(x,\cdot))_{L^2(I)}$ and $\sharp = D, N$. We use the fact that all $|A_j|$ can be estimated by a constant $c$ depending only on $|\alpha_0|$ and $d$, and Parseval identity for $\chi_j^D$ and $\chi_j^N$ to estimate
\beq \label{suma-pro-dalsi-lemma}
\bal
 \sum_{j=2}^{+\infty} |\Psi_j(x)|^2 \leq c^2 d \left( \sum_{j=2}^{+\infty}|\Psi_j^N|^2 + \frac{\alpha_0^2}{\mu_2^2}\sum_{j=2}^{+\infty}|\Psi_j^D|^2\right) \leq c^2 d \left(1 +\frac{\alpha_0^2}{\mu_2^2}\right) \|\Psi(x,\cdot)\|^2_{L^2(I)}.
\eal
\eeq
To estimate the first two terms in the sum \eqref{bi-expansion} we simply use the inequality $|\psi_j|^2 \leq 1 + \frac{\alpha_0^2}{\mu_0^2}$ valid for all $j \geq 0$. We get
\beq
\bal
\left\|\sum_{j=0}^1 \Psi_j(x) \psi_j\right\|^2_{L^2(I)} \leq 2 c^2 d^2\left(1 + \frac{\alpha_0^2}{\mu_0^2}\right)^2 \|\Psi(x,\cdot)\|^2_{L^2(I)}.
\eal
\eeq
Altogether we have
\beq
\left\|\sum_{j=0}^k \Psi_j(x_1) \psi_j\right\|_{L^2(I)} \leq C \|\Psi(x,\cdot)\|_{L^2(I)} \in L^2(\RR^n).
\eeq
Since the constant $C$ does not depend on $k$, we are provided with a uniform estimate and the sum \eqref{bi-expansion} converges to $\Psi$ in $L^2(\Omega)$ norm.
\end{proof}
%
\subsection{Spectrum of the unperturbed Hamiltonian}
%
We aim to proof Proposition \ref{noper}. It is quite straightforward to see that its point spectrum is empty under the hypothesis (\ref{simple-spectrum}), i.e.
\begin{lem} \label{emptypoint}
Let $\alpha_0$ satisfy $\eqref{simple-spectrum}$. Then
\beq
\sigma_{\mathrm{p}}(H_{\alpha_0}) = \emptyset
\eeq
\end{lem}
\begin{proof}
For the contradiction let us assume that $H_{\alpha_0}$ possesses an eigenvalue $\lambda$ with an eigenfunction $\Psi \in L^2(\Omega)$. We then multiply the eigenvalue equation with $\overline{\phi_j}$ and integrate it over $I$. Adopting the notation $\Psi_j(x) := \left(\phi_j,\Psi(x,\cdot)\right)_{L^2(I)}$ the equation then reads
\beq \label{eigenproblem-in-R}
-\Psi_j'' = (\lambda - \mu_j^2)\Psi_j
\eeq
in $\RR^n$ for every $j \geq 0$. Using Schwarz inequality and Fubini's theorem we see that $\Psi_j \in L^2(\RR^n)$:
\beq
\bal \label{dalsi-odhad-psi}
\|\Psi_j\|_{L^2(\RR^n)}^2 \leq \int_{\RR^n}\|\phi_j(u)\|_{L^2(I)}^2 \|\Psi(x,u)\|_{L^2(I)}^2 \dx = \|\phi_j\|_{L^2(I)}^2 \|\Psi\|_{L^2(\Omega)}^2 < + \infty.
\eal
\eeq
Since the point spectrum of the Laplacian in $\RR^n$ is empty, equation (\ref{eigenproblem-in-R}) only has a trivial solution. Therefore, (\ref{bi-expansion}) yields $\Psi = 0$, which is in contradiction with our hypothesis.
\end{proof}
\begin{rmrk} \label{extend-point}
We can further claim that the set of isolated eigenvalues is always empty, even in the case when the condition $\eqref{simple-spectrum}$ is not satisfied. This is the consequence of the fact that $H_{\alpha_0}$ forms a holomorphic family of operators of type (B) with respect to $\alpha_0$ and hence all its isolated eigenvalues $\mu_j(\alpha_0)$ are analytic functions in $\alpha_0$ \cite[Sec. VII.4]{Kato}. 
\end{rmrk}
The essential spectrum behaves, as can be expected - it consists of the essential spectrum of the free Laplacian in $\RR^n$, shifted by the lowest-lying eigenvalue of $-\Delta^I$.
\begin{lem} \label{ess-lem1}
Let $\alpha_0 \in \RR$. Then $[\mu_0^2, +\infty) \subset \sigma_{\mathrm{ess}}(H_{\alpha})$
\end{lem}
\begin{proof}
Let $\lambda \in [\mu_0^2,+\infty)$. It can be expressed as $\lambda = \mu_0^2 + z$, where $z \in [0,+\infty$). Let $\left(\Phi_k\right)_{k=1}^{+\infty} \subset L^2(\RR^n)$ be a singular sequence of $-\Delta'$ corresponding to $z$, i.e. $\|\Phi_k\|_{L^2(\RR^n)} = 1$, $\left(\Phi_k\right)_{k=1}^{+\infty}$ does not contain converging subsequence and $(-\Delta' - z)\Phi_k \rightarrow 0$. We define sequence $(\Psi_k)_{k=1}^{+\infty}$ by $\Psi_k(x,u) := \Phi_k(x) \psi_0(u) / \|\psi_0\|_{L^2(I)}$. It can be easily seen that $\|\Psi_k\|_{L^2(\Omega)} = 1$ for all $k \in \NN$ and $\Psi_k \rightarrow 0$ and that $(H_{\alpha_0} - \lambda)\Psi_k \rightarrow 0$ since
\beq
\bal
(H_{\alpha_0} - z - \mu_0^2)\Psi_k =\left((-\Delta' - z) \Phi_k\right)\psi_0/\|\psi_0\|_{L^2(I)} \rightarrow 0.
\eal
\eeq
In other words, $(\Psi_k)_{k=1}^{+\infty}$ forms a singular sequence for $\lambda$ and it is therefore part of the essential spectrum.
\end{proof}
The opposite inclusion can be seen by employing the decomposition of the resolvent into the transverse biorthonormal basis.
\begin{lem} \label{ess-lem2}
Let $\alpha_0$ satisfy \eqref{simple-spectrum}. Then $\CC \setminus [\mu_0^2,+\infty) \subset \rho(H_{\alpha_0})$ and for any $\lambda \in \CC \setminus [\mu_0^2, +\infty)$ we have
\beq \label{rezolvent-suma}
(H_{\alpha_0} - \lambda)^{-1} = \sum_{j=0}^{+\infty} (-\Delta' + \mu_j^2 - \lambda)^{-1} B_j.
\eeq
Here $B_j$ is a bounded operator on $L^2(\Omega)$ defined by 
\beq
\left(B_j \Psi\right)(x,u) := \left(\Psi(x,\cdot),\phi_j\right)_{L^2(\Omega)} \psi_j(u)
\eeq
for $\Psi \in L^2(\Omega)$ and $(-\Delta' + \mu_j^2 - \lambda)^{-1} $ abbreviates $(-\Delta' + \mu_j^2 - \lambda)^{-1} \otimes \JJ$.
\end{lem}
\begin{proof}
We proceed with the proof as in \cite[Lem. 4.3]{Borisov08}. Let $\lambda \in \CC\setminus [\mu_0^2, +\infty)$ and $\Psi \in L^2(\Omega)$. We denote $\Psi_j(x) := (\phi_j,\Psi(x,\cdot))_{L^2(I)} \in L^2(\RR^n)$ and $U_j := (-\Delta' + \mu_j^2 - \lambda)^{-1} \Psi_j \in L^2(\RR^n)$ for $j\geq 0$. Its norm can be estimated as
\beq \label{odhad-U}
\bal
\|U_j\|_{L^2(\RR^n)} &\leq \frac{\|\Psi_j\|_{L^2(\RR^n)}}{\mathrm{dist}(\lambda,[\mu_j^2,+\infty))} \leq C_1 \frac{\|\Psi_j\|_{L^2(\RR^n)}}{j^2 + 1}.
\eal
\eeq
The constant $C_1$ depends only on $|\alpha_0|$, $d$ and $\lambda$. Similarly, we estimate $|\partial_{x_l} U_j|$ for every $j\geq 0$, $l \geq 1$ by its gradient in $\RR^n$ and we obtain 
\beq
\bal \label{odhad-UU}
\|\nabla' U_j\|_{L^2(\RR^n)}^2 \leq C_1 \frac{\|\Psi_j\|_{L^2(\RR^n)}^2}{j^2 +1} + C_1^2 |\mu_j^2 - \lambda| \frac{\|\Psi_j\|_{L^2(\RR^n)}^2}{(j^2 + 1)^2}.
\eal
\eeq
We define a function $R_j(x) := U_j(x) \psi_j(u)$ (which is exactly the summand of the sum \eqref{rezolvent-suma}). It belongs to $W^{2,1}(\Omega)$ and this is true for their infinite sum too as we shall see. Employing estimates \eqref{suma-pro-dalsi-lemma} and \eqref{odhad-U} together with Fubini's theorem yields
\beq
\bal
\left\|\sum_{j=2}^k R_j \right\|^2_{L^2(\Omega)} 
& \leq d^2 \left( 1 + \frac{\alpha_0^2}{\mu_2^2}\right)^2 C_1^2 \sum_{j=2}^k \frac{\|\Psi_j\|_{L^2(\RR^n)}}{(j^2 + 1)^2}\\
&\leq d^2 \left( 1 + \frac{\alpha_0^2}{\mu_2^2}\right)^2 C_1^2 \int_{\RR^n} \sum_{j=2}^k |(\phi_j,\Psi(x,\cdot))_{L^2(I)}|^2 \dx \\
& \leq c^2 d^3 \left( 1 + \frac{\alpha_0^2}{\mu_2^2}\right)^3 C_1^2 \|\Psi\|_{L^2(\Omega)}.
\eal
\eeq
We remind that constant $c$ depends only on $|\alpha_0|$, $d$ and $\lambda$, just as $C_1$. In exactly the same manner we estimate $\|\sum_{j=1}^k \partial_{x_l}R_j\|_{L^2(\Omega)}$ for every $l\geq 1$ using the estimate \eqref{odhad-UU} instead of \eqref{odhad-U}. Employing the estimate $|\partial_u\psi_j| \leq \alpha_0^2 + \mu_j^2$ valid for $j\geq 1$, we readily estimate the norm of $\sum_{j=2}^k \partial_u R_j$:
\beq
\bal
\left\|\sum_{j=1}^k \partial_u R_j\right\|^2_{L^2(\Omega)} 
&\leq d^2  C_1^2 \sum_{j=2}^k \left(\frac{\mu_j^2 + \alpha_0^2}{j^2 + 1}\right)^2 \|\Psi_j\|_{L^2(\RR^n)}^2\\
&\leq d^2  C_1^2 C_2^2 \int_{\RR^n} \sum_{j=2}^k |(\phi_j,\Psi(x,\cdot))_{L^2(I)}|^2 \dx \\
&\leq c^2 d \left( 1 + \frac{\alpha_0^2}{\mu_2^2}\right) C_1^2 C_2^2 \|\Psi\|_{L^2(\Omega)},
\eal
\eeq
where $C_2$ is a constant bounding the sequence $\left(\frac{\mu_j^2 + \alpha_0^2}{j^2 + 1}\right)_{j=2}^{+\infty}$, depending only on $|\alpha_0|$ and $d$. Regarding the sum of the first two terms, we obtain
\beq
\bal
\left\|\sum_{j=0}^1 R_j\right\|^2_{L^2(\Omega)} \leq d^2\left(1 + \frac{\alpha_0^2}{\mu_0^2}\right)^2 C_1^2 \sum_{j=0}^1 \frac{\|\Psi_j\|^2_{L^2(\RR^n)}}{(j^2 + 1)^2} \leq c^2 d^3\left(1 + \frac{\alpha_0^2}{\mu_0^2}\right)^3 C_1^2 \|\Psi\|_{L^2(\Omega)}
\eal
\eeq
and similarly for $\partial_{x_l}R_j$ and $\partial_u R_j$. Altogether we uniformly estimated the partial sum of $R_j$ and of its derivatives, and therefore the series $\sum_{j=0}^{+\infty}R_j$ converges in $W^{1,2}(\Omega)$ to a function $R$ and 
\beq
\|R\|_{W^{1,2}(\Omega)} \leq K \|\Psi\|_{L^2(\Omega)},
\eeq
where $K$ depends only on $|\alpha_0|$, $d$ and $\lambda$. It is easily seen that $R$ satisfies the identity 
\beq
h_{\alpha_0}(R,\Phi) - \lambda(R,\Phi)_{L^2(\Omega)} = (\Psi,\Phi)_{L^2(\Omega)}
\eeq
for all $\Phi \in W^{1,2}$. Therefore, $R \in \Dom(H_{\alpha_0})$ and $(H_{\alpha_0} - \lambda) R = \Psi$, i.e. $R = (H_{\alpha_0} - \lambda)^{-1}\Psi$.
\end{proof}
\begin{proof}[Proof of Proposition \ref{noper}]
The first inequality follows from Lemma \ref{emptypoint}. From Lemmas \ref{ess-lem1} and \ref{ess-lem2} we know that the second equality holds for all $\alpha_0$ satisfying \eqref{simple-spectrum}. This result extends to all $\alpha_0$ in view of the fact that $H_{\alpha_0}$ forms a holomorphic family of operators of type (B) with respect to $\alpha_0$ (cf. Remark \ref{extend-point}).
\end{proof}
%
\subsection{Stability of the essential spectrum}
%
Our goal is to find conditions under which a single bound state arises as a consequence of a perturbation of the boundary conditions. Generally, it could happen that although it appears, the essential spectrum changes in such a way that it is absorbed in it. Therefore, we first investigate the stability of the essential spectrum under perturbations of uniform boundary conditions studied in detail in previous section and conclude with the proof of Theorem \ref{per}. Let us state an auxiliary lemma.
\begin{lem} \label{stability-lemma}
Let $\alpha_0 \in \RR$ and $\varphi \in L^2(\partial\Omega)$. There exist positive constants $c$ and $C$, depending on $d$ and $|\alpha_0|$, such that any weak solution $\Psi \in W^{1,2}(\Omega)$ of the boundary value problem
\beq \label{stability-equation}
\begin{dcases}
(-\Delta - \lambda)\Psi \!\!\!&= 0 \qquad \mathrm{in}\quad \Omega,\\
(\partial_u + \ii \alpha_0) \Psi \!\!\!&= \varphi \qquad \mathrm{on} \quad \partial\Omega,
\end{dcases}
\eeq
with any $\lambda\leq -c$, satisfies the estimate
\beq \label{nejaky-odhad}
\|\Psi\|_{W^{1,2}(\Omega)} \leq C \|\varphi\|_{L^2(\partial \Omega)}.ě
\eeq
\end{lem}
\begin{proof}
Multiplying the first equation of \eqref{stability-equation} by $\overline{\Psi}$ and integrating over $\Omega$ yields
\beq
\bal \label{one-use}
\int_{\Omega}\overline{\Psi(x,u)}(-\Delta-\lambda)\Psi(x,u) \dx\du &= \ii \int_{\RR^n} \alpha_0 |\Psi(x,d)|^2 dx - \ii \int_{\Omega} \alpha_0 |\Psi(x,0)|^2 \dx \\
& - \int_{\RR^n} \overline{\Psi(x,d)}\varphi(x,d) \dx + \int_{\RR^n} \overline{\Psi(x,0)}\varphi(x,0) \dx \\
& + \|\nabla\Psi\|_{L^2(\Omega)}^2 - \lambda\|\Psi\|_{L^2(\Omega)}^2 = 0
\eal
\eeq
We readily estimate using Schwarz and Young inequality
\beq
\bal
\left| \int_{\RR^n} \alpha_0 |\Psi(x,d)|^2 dx - \int_{\Omega} \alpha_0 |\Psi(x,0)|^2 \dx \right| &= \left|\int_{\Omega} \alpha_0\partial_u |\Psi(x)|^2 \dx\right| \\
&= 2 \left|\alpha_0\right|\left| \Re (\partial_u \Psi, \Psi)\right| \\
&\leq 2 \left|\alpha_0\right| \|\partial_u \Psi\|_{L^2(\Omega)} \|\Psi\|_{L^2(\Omega)}  \\
& \leq \left|\alpha_0\right| \left(\delta \|\nabla \Psi\|_{L^2(\Omega)}^2 + \delta^{-1} \|\Psi\|_{L^2(\Omega)}^2\right)
\eal
\eeq
and
\beq
\bal
\left| - \int_{\RR^n} \overline{\Psi(x,d)}\varphi(x,d) \dx + \int_{\RR^n} \overline{\Psi(x,0)}\phi(x,0) \dx \right| &\leq
2 \|T \Psi\|_{L^2(\partial \Omega)} \|\varphi\|_{L^2(\partial \Omega)} \\
&\leq \delta \widetilde C \|\Psi\|_{W^{1,2}(\Omega)}^2 + \delta^{-1} \|\varphi\|_{L^2(\partial \Omega)}^2,
\eal
\eeq
where $\delta >0$ and $\widetilde C$ is the constant from the embedding of $W^{1,2}(\Omega)$ in $L^2(\Omega)$ depending only on $d$. Putting these estimates into \eqref{one-use} we get
\beq
\left(1 - \delta |\alpha_0| - \delta \widetilde C\right) \|\Psi\|_{W^{1,2}(\Omega)}^2 \leq \left(1-\delta |\alpha_0| + \delta^{-1} |\alpha_0| +\lambda\right)\|\Psi\|_{L^2(\Omega)}^2+ \delta^{-1} \|\varphi\|_{L^2(\partial \Omega)}^2 
\eeq
Taking $\delta$ sufficiently small and $\lambda$ sufficiently large negative, coefficients standing by $\|\Psi\|_{W^{1,2}(\Omega)}$ and $\|\Psi\|_{L^2(\Omega)}$ are positive and this yields the inequality \eqref{nejaky-odhad}.
\end{proof}
Using this lemma we are able to prove the following result.
\begin{prop} \label{compres}
Let $\alpha - \alpha_0 \in W^{1,\infty}(\RR)$ with $\alpha_0 \in \RR$ such that \eqref{constantinfinity} holds. Then $(H_{\alpha}-\lambda)^{-1} - (H_{\alpha_0} - \lambda)^{-1}$ is compact in $L^2(\Omega)$ for any $\lambda \in \rho(H_{\alpha})\cap\rho(H_{\alpha_0})$.
\end{prop}
\begin{proof}
The proof is inspired by the proof of \cite[Prop. 5.1]{Borisov08}. It suffices to prove the result only for one $\lambda \in \rho(H_{\alpha})\cap\rho(H_{\alpha_0})$ sufficiently negative. (Since both $H_{\alpha_0}$ and $H_{\alpha}$ are m-sectorial, their spectra are bounded from below.) The result can be then extended to any other $\lambda' \in \rho(H_{\alpha})\cap\rho(H_{\alpha_0})$ due to the first resolvent identity. Let us denote for this purpose $R(H_{\alpha};\lambda):=(H_{\alpha} - \lambda)^{-1}$ and $R(H_{\alpha_0};\lambda):=(H_{\alpha_0} - \lambda)^{-1}$. Then we have
\beq
\bal
R(H_{\alpha};\lambda') &- R(H_{\alpha_0};\lambda') \\
 & = R(H_{\alpha};\lambda')\left(\JJ + (\lambda'-\lambda)R(H_{\alpha};\lambda')\right) - \left(\JJ + (\lambda'-\lambda)R(H_{\alpha};\lambda')\right)R(H_{\alpha_0};\lambda')\\
& =\left(\JJ + (\lambda'-\lambda)R(H_{\alpha};\lambda')\right)\left(R(H_{\alpha};\lambda) - R(H_{\alpha_0};\lambda)\right)\left(\JJ + (\lambda'-\lambda)R(H_{\alpha};\lambda')\right).\\
\eal
\eeq 
From the assumption $R(H_{\alpha};\lambda) - R(H_{\alpha_0};\lambda)$ is compact and $\JJ + (\lambda'-\lambda)R(H_{\alpha};\lambda')$ and $\JJ + (\lambda'-\lambda)R(H_{\alpha};\lambda')$ are bounded. The claim then follows from the two side ideal property of compact operators. Given an arbitrary $\Phi \in L^2(\Omega)$, let us define $\Psi := (H_{\alpha} - \lambda)^{-1}\Phi - (H_{\alpha_0} - \lambda)^{-1}\Phi$. $\Psi$ clearly satisfies the first equation in (\ref{stability-equation}). Plugging it into the second one we get
\beq
\bal
(\partial_2 + \ii \alpha_0)\Psi = (\partial_2 + \ii \alpha_0)\left((H_{\alpha} - \lambda)^{-1}\Phi - (H_{\alpha_0} - \lambda)^{-1}\Phi \right)= - \ii (\alpha - \alpha_0) (H_{\alpha} - \lambda)^{-1} \Phi,
\eal
\eeq
therefore our $\varphi = - \ii (\alpha - \alpha_0) T (H_{\alpha} - \lambda)^{-1} \Phi$, where $T$ is a trace operator from $W^{2,2}(\Omega)$ to $W^{1,2}(\partial \Omega)$. Due to the estimate \ref{nejaky-odhad} it is enough to show that $(\alpha-\alpha_0)T(H_{\alpha} - \lambda)^{-1}$ is compact. Indeed if this is true then given any sequence $(\Phi_n)_{n=1}^{+\infty} \subset L^2(\Omega)$ we know there is a strictly increasing sequence $(k_n)_{n=1}^{+\infty} \subset \NN$ such that for every $\varepsilon > 0$ there is $n_0 \in \NN$ such that for all $m,n > n_0$ inequality $\left\|\left((\alpha-\alpha_0)T(H_{\alpha} - \lambda)^{-1}\right)(\Phi_m - \Phi_n)\right\|<\varepsilon$ holds. It follows the same is true for $(H_{\alpha} - \lambda)^{-1} - (H_{\alpha_0} - \lambda)^{-1}$ since
\beq
\bal
\left\|\left(H_{\alpha}-\lambda\right)^{-1}\right.&\left.- \left(H_{\alpha_0}-\lambda\right)^{-1})(\Phi_m-\Phi_n)\right\|_{L^2(\Omega)} \leq C \left\|\left((\alpha-\alpha_0)T(H_{\alpha} - \lambda)^{-1}\right)(\Phi_m - \Phi_n)\right\|_{L^2(\Omega)}.
\eal
\eeq
We denote $\beta := \alpha -  \alpha_0$ and define functions
\beq
\beta_n(x):=
\begin{dcases}
\beta(x), \quad x\in (-n,n) \\
0 \quad \mathrm{otherwise}.\\
\end{dcases}
\eeq
These bounded continuous functions with compact support converge to $\beta(x)$ in $L^{\infty}(\RR^n)$ norm. $\beta_n T(H_{\alpha} - \lambda)^{-1}$ is a compact operator since $W^{1,2}(\partial\Omega)$ is compactly embedded in $L^2(\omega)$ for every bounded subset $\omega$ of $\partial\Omega$, due to the Rellich-Kondrachov theorem \cite[Sec. VI]{Adams1975}. In other words, every set $A$, which is bounded in the topology of $W^{1,2}(\partial\Omega)$, is precompact in the topology of $L^2(\omega)$. The claim then follows from the two sided ideal property of the set compact operator if we show that the compact operators $\beta_n T (H_{\alpha} - \lambda)^{-1}$ converge in the uniform $L^{2}(\partial\Omega)$ topology to our operator $\beta_n T (H_{\alpha} - \lambda)^{-1}$. We have
\beq
\|\beta T (H_{\alpha} - \lambda)^{-1} - \beta_n T (H_{\alpha} - \lambda)^{-1} \| \leq \|\beta - \beta_n\|_{L^{\infty}(\RR^n)} \|T (H_{\alpha} - \lambda)^{-1} \|,
\eeq
which converges to $0$ for $n\rightarrow +\infty$. 
\end{proof}
\begin{proof}[Proof of Theorem \ref{per}]
Since the difference of the resolvents is a compact operator according to Proposition \ref{compres}, it follows from the Weyl's essential spectrum theorem \cite[Thm. XIII.14]{Reed4} that the essential spectra of $H_{\alpha}$ and $H_{\alpha_0}$ are identical.
\end{proof}
%
\section{Weakly coupled bound states} \label{sekcesest}
%
Another possible influence of the perturbation of the boundary conditions on the spectrum is studied in this section. We shall employ the form 
\beq
\alpha(x) = \alpha_0 + \varepsilon \beta(x)
\eeq
for $\alpha$ further on. Here $\alpha_0 \in \RR$, $\beta \in W^{2,\infty}(\RR^n)$ and $\varepsilon > 0$. This section contains some preliminary and auxilliary results and culminates with the proof of Theorem \ref{vetavet}
%
\subsection{Unitary transformation of \texorpdfstring{$H_{\alpha}$}{H alpha}}
%
The form $\eqref{H-waveguide}$ is not very convenient for the study of bound states, the unitary transformation is therefore applied to simplify the boundary conditions for the cost of an adding of a differential operator.
\begin{prop}
$H_{\alpha}$ is unitarily equivalent to the operator $H_{\alpha_0} + \varepsilon Z_{\varepsilon}$, where
\beq
Z_{\varepsilon} := 2 \ii u \nabla'\beta(x) \cdot \nabla' + 2 \ii \beta(x)\frac{\partial}{\partial u} + \left(\varepsilon\beta^2(x) + \ii \Delta' \beta(x)u + \varepsilon u^2 |\nabla' \beta|^2\right).
\eeq
and $\Dom(H_{\alpha_0} + \varepsilon Z_{\varepsilon}) = \Dom(H_{\alpha_0})$.
\end{prop}
\begin{proof}
We are going to show that the relation
\beq
U_{\varepsilon}^{-1}H_{\alpha} U_{\varepsilon} = H_{\alpha_0} + \varepsilon Z_{\varepsilon},
\eeq
holds in operator sense with the unitary operator of multiplication $U_{\varepsilon}$ acting on $\Psi \in L^2(\Omega)$ as $\left(U_{\varepsilon}\Psi\right)(x,u):=\ee^{-\ii \varepsilon \beta(x) u}\Psi(x,u)$. First we show that $\Dom(U_{\varepsilon}^{-1}H_{\alpha} U_{\varepsilon}) = \Dom(H_{\alpha_0} + \varepsilon Z_{\varepsilon})$. Simple calculations show that $\Dom(H_{\alpha_0}) = \Dom(U_{\varepsilon}^{-1}H_{\alpha} U_{\varepsilon})$. Further, $U_{\varepsilon}^{-1}H_{\alpha} U_{\varepsilon}$ and $ H_{\alpha_0} + \varepsilon Z_{\varepsilon}$ act in the same on functions from their domain. Now we prove that $\Dom(H_{\alpha_0} + \varepsilon Z_{\varepsilon}) = \Dom(H_{\alpha_0})$. It is clear that domain of $H_{\alpha_0} + \varepsilon Z_{\varepsilon}$ is a subset of the domain of $H_{\alpha_0}$. Taking $\Psi \in \Dom(H_{\alpha_0}) \subset W^{2,2}(\Omega)$ we estimate every action of $Z_{\varepsilon}$ as
\beq
\bal
\|2 \ii u (\nabla' \beta)(x)\cdot \nabla'\Psi\|_{L^2(\Omega)} &\leq 2 n d \|\beta\|_{W^{2,\infty}(\RR^n)} \|\Psi\|_{W^{2,2}(\Omega)},\\
\|2 \ii \beta \partial_u \Psi\|_{L^2(\Omega)} &\leq 2\|\beta\|_{W^{2,\infty}(\RR^n)} \|\Psi\|_{W^{2,2}(\Omega)}, \\
\|\varepsilon \beta^2 \Psi\|_{L^2(\Omega)} &\leq \varepsilon \|\beta\|_{W^{2,\infty}(\RR^n)}^2 \|\Psi\|_{W^{2,2}(\Omega)},\\
\|-\ii \Delta' \beta u \Psi\|_{L^2(\Omega)} &\leq d \|\beta\|_{W^{2,\infty}(\RR^n)} \|\Psi\|_{W^{2,2}(\Omega)},\\
\|-\varepsilon u^2 |\nabla' \beta|^2 \Psi\|_{L^2(\Omega)} &\leq \varepsilon d^2 \|\beta\|_{W^{2,\infty}(\RR^n)}^2 \|\Psi\|_{W^{2,2}(\Omega)}^2.\\
\eal
\eeq
In other words we just showed that $\Dom(H_{\alpha_0}) \subset \Dom(H_{\alpha_0} + \varepsilon Z_{\varepsilon})$ and the equality of domains is proven. 
\end{proof}
Overall, we were able to transform away the perturbed boundary conditions at the cost of adding a differential operator to the unperturbed Hamiltonian. Since unitarily equivalent operators possesses identical spectra, further on we are going to study the operator $H_{\alpha_0} + \varepsilon Z_{\varepsilon}$. Hereafter, a straightforward calculation inspired by \cite{BEGK} proves that
\beq \label{prvni-rozklad}
Z_{\varepsilon} = \sum_{i=1}^{n+2}A_i^*B_i + \varepsilon \sum_{i=n+3}^{2n+3}A_{i}^*B_{i},
\eeq
where $A_i$ and $B_i$ are first-order differential operators, specifically
\beq
\bal
A_1^* &:= 2\ii\left(\partial_{x_1}\beta(x)\right)_{1/2}u, & B_1 &:= \left|\partial_{x_1}\beta(x)\right|^{1/2}\frac{\partial}{\partial x_1}, \\
\vdots &   &  &\vdots \\
A_n^* &:= 2 \ii \left(\partial_{x_n} \beta(x)\right)_{1/2} u, & B_n &:= \left|\partial_{x_n} \beta(x)\right|^{1/2}\frac{\partial}{\partial x_n}, \\
A_{n+1}^* &:= 2\ii\beta(x)_{1/2},     & B_{n+1} &:= |\beta(x)|^{1/2}\frac{\partial}{\partial u}, \\
A_{n+2}^* &:= -\ii \left(\Delta'\beta(x)\right)_{1/2} u,& B_{n+2} &:= \left|\Delta'\beta(x)\right|^{1/2}, \\
A_{n+3}^* &:= \beta(x) u^2, & B_{n+3} &:= \beta(x),\\
A_{n+4}^* &:= \partial_{x_1}\beta(x) u^2, & B_{n+4} &:= \partial_{x_1}\beta(x),\\
\vdots & & & \vdots \\
A_{2n+3}^* &:= \partial_{x_n}\beta(x) u^2, & B_{2n+3} &:= \partial_{x_n}\beta(x),
\eal
\eeq
where $\left(f(x)\right)_{1/2} := \mathrm{sgn}\left(f(x)\right) \left|f(x)\right|^{1/2}$ for any function $f$. We define a pair of operators $C_{\varepsilon}, D: L^2(\Omega) \rightarrow L^2(\Omega) \otimes \CC^{2n+3}$ by
\beq
\bal
(C_{\varepsilon} \varphi)_i &:=
\begin{dcases}
A_i \varphi, \qquad i = 1,2,\cdots, n+2,\\
\varepsilon A_i \varphi, \qquad i=n+3,\cdots,2n+3, 
\end{dcases} \\
(D\varphi)_i &:= B_i \varphi, \qquad i=1, \dots, 2n+3.
\eal
\eeq
Then \eqref{prvni-rozklad} finally becomes.
\beq
U_{\varepsilon}^{-1} H_{\alpha} U_{\varepsilon} = H_{\alpha_0} + \varepsilon C_{\varepsilon}^* D.
\eeq
(Note that $C_{\varepsilon}^*$ project from $L^2(\Omega) \otimes \CC^{2n+3}$ to $L^2(\Omega)$ according to the definition of adjoint operator.) 
%
\subsection{Birman-Schwinger principle}
%
We introduce a useful technique for studying certain types of partial differential equations, particularly in the analysis of the point spectrum of differential operators. It was developed independently by Russian mathematician M. Sh. Birman \cite{birman61} and American physicist J. Schwinger \cite{schwinger61} in the year 1961 for estimating the number of negative eigenvalues of a self-adjoint Schr\"{o}dinger operator. Since its origin it was applied in finding weakly coupled bound states \cite{Simon-76}, studying behaviour of the resolvent \cite{klaus82}, localizing the spectrum \cite{Davies2014} and also finding eigenvalue bounds in non-self-adjoint operators \cite{davies01, frank11, laptev09}. Generally it enables us to solve an eigenvalue problem for differential operators by solving an eigenvalue problem for integral operators. In this paper we apply it on the non-self-adjoint operator. Since $Z_{\varepsilon}$ is a differential operator, we will have to employ regularity of functions involved and integration by parts to obtain an integral operator (cf. proof of Lemmas \ref{lemma3} and \ref{lemma5}).
\begin{prop}
Let $\lambda \in \CC\setminus[0,+\infty)$, $\varepsilon\in \RR$, $\beta \in W^{2,\infty}(\RR^n)$ such that 
\beq
\lim_{|x|\rightarrow +\infty} \beta(x) = \lim_{|x|\rightarrow +\infty} \partial_{x_j}\beta(x) = \lim_{|x|\rightarrow +\infty} \partial_{x_j}^2\beta(x)=0
\eeq
for all j = 1, \dots, n. Denoting $K_{\varepsilon}^{\lambda} := \varepsilon D (H_{\alpha_0} - \lambda)^{-1} C_{\varepsilon}^*$, then
\beq
\lambda \in \sigma_{\mathrm{p}}(H_{\alpha}) \quad \Leftrightarrow \quad -1 \in \sigma_{\mathrm{p}}(K_{\varepsilon}^{\lambda}).
\eeq
\end{prop}
\begin{proof}
$\Rightarrow:$ Assuming $H_{\alpha} \Phi = \lambda \Phi$ holds for some $\Phi \in \Dom(H_{\alpha})$ we define $\Psi := D \Phi$. $\Psi \in L^2(\Omega) \otimes \CC^{2n+3}$ since we have for each $(D\Phi)_i$ the following estimate:
\beq
\bal
\|(D \Phi)_i\|_{L^2(\Omega)} \leq c_1\|\nabla \Phi\|_{L^2(\Omega)} + c_2\|\Phi\|_{L^2(\Omega)} \leq (c_1 + c_2)\|\Phi\|_{W^{2,2}(\Omega)} < + \infty,
\eal
\eeq
where $c_1$ and $c_2$ are constants arising from the boundedness of $\beta$, its derivatives and their square roots (cf. \eqref{odhad-D}). For this $\Psi$ we then have
\beq
\bal
K_{\varepsilon}^{\lambda} \Psi = \varepsilon D (H_{\alpha_0} - \lambda)^{-1} C_{\varepsilon}^* D \Phi= - D (H_{\alpha_0} - \lambda)^{-1} (H_{\alpha_0} - \lambda) \Phi  = - \Psi.
\eal
\eeq
$\Leftarrow:$ Let us assume that $\Psi \in  L^2(\Omega) \otimes \CC^{2n+3}$ is an eigenfunction of $K_{\varepsilon}^{\lambda}$ pertaining to the eigenvalue $-1$. The assumptions imply that $\beta, \partial_{x_j}\beta$ and $\partial_{x_j}^2\beta$ are bounded for all $j = 1, \dots n$, therefore the operator $C_{\varepsilon}$ is bounded and the same applies for its adjoint (cf. \eqref{odhad-C}). Then $\Phi := - (H_{\alpha_0} - \lambda)^{-1}C_{\varepsilon}^* \Psi \in W^{2,2}(\Omega)$ and 
\beq
\bal
(H_{\alpha_0} - \lambda)\Phi &= (H_{\alpha_0} - \lambda)(H_{\alpha_0} - \lambda)^{-1}C_{\varepsilon}^* \Psi =- \varepsilon C_{\varepsilon}^* D (H_{\alpha_0} - \lambda)^{-1} C_{\varepsilon}^*\Psi = - \varepsilon C_{\varepsilon}^* D \Phi.
\eal
\eeq
\end{proof}
%
\subsection{Structure of \texorpdfstring{$K_{\varepsilon}^{\lambda}$}{K epsilon lambda}}
%
To analyze the structure of $K_{\varepsilon}^{\lambda}$ we take a closer look on the resolvent operator $(H_{\alpha_0}- \lambda)^{-1}$. We have shown in Lemma \ref{ess-lem2} that the biorthonormal-basis-type relations (\ref{biortho-relation}) enable us to decompose the resolvent of $H_{\alpha_0}$ into the transverse biorthonormal-basis. Its integral kernel then for every $\lambda \in \CC \setminus [\mu_0^2, + \infty)$ reads
\beq \label{resolvent-kernel}
\left((H_{\alpha_0} - \lambda)^{-1}\right)(x,u,x',u') = \sum_{j=0}^{+\infty} \psi_j(u) \mathcal{R}_{\mu_j^2 - \lambda}(x,x') \overline{\phi_j(u')} ,
\eeq 
where $\psi_j$ and $\phi_j$ were defined in (\ref{function}) and (\ref{adjoint-function}), respectively, and $\mathcal{R}_{\mu_j^2 - \lambda}(x,x')$ is the integral kernel of $(-\Delta' + \mu_j^2 - \lambda)$. This naturally differs for various ``longitudal'' dimensions $n$.  It is an integral operator with the integral kernel
\beq \label{reduced-kernel}
\mathcal{R}_{z}(x,x') =
\begin{dcases}
\frac{\ee^{-\sqrt{-z}|x-x'|}}{2\sqrt{-z}} &\mbox{if} \quad n=1,\\
\frac{1}{2\pi} K_0(\sqrt{-z}|x-x'|) & \mbox{if} \quad n = 2,\\
\frac{\ee^{-\sqrt{-z}|x-x'|}}{4 \pi |x-x'|}& \mbox{if} \quad n \geq 3.
\end{dcases}
\eeq
Here $K_0$ is Macdonald's function \cite[9.6.4]{Abramowitz}. In this paper we are interested only in the case $n = 1,2$ (cf. Remark \ref{singularita-rmrk}). To study $K_{\varepsilon}^{\lambda}$ it is necessary to somehow deal with the singularity arising from the first term in the sum (\ref{resolvent-kernel}) when $\lambda$ tends to $\mu_0^2$. Hence, following \cite{Simon-76} we decompose it into two operators, $K_{\varepsilon}^{\lambda} = L_{\varepsilon}^{\lambda} + M_{\varepsilon}^{\lambda}$, separating the diverging part in the operator, $L_{\varepsilon}^{\lambda} := \varepsilon D L_{\lambda} C_{\varepsilon}^*$, where $L_{\lambda}$ is an integral operator with the kernel
\beq
\mathcal{L}_{\lambda}(x,u,x',u') := 
\begin{dcases}
\psi_0(u) \frac{1}{2\sqrt{\mu_0^2 - \lambda}}\, \overline{\phi_0(u')} & \mbox{if} \quad n = 1,\\
- \frac{1}{2\pi}\psi_0(u) \ln\sqrt{\mu_0^2 - \lambda} \,\overline{\phi_0(u')} &\mbox{if} \quad n = 2.
\end{dcases}
\eeq
We see that indeed the integral kernel of $L_{\lambda}$ diverges for $\lambda$ tending to $\mu_0^2$. For technical reasons the regular part $M_{\varepsilon}^{\lambda} := \varepsilon D M_{\lambda} C_{\varepsilon}^*$ can again be divided into two terms, $N_{\varepsilon}^{\lambda} := \varepsilon D N_{\lambda} C_{\varepsilon}^*$ and $\varepsilon D R_{\alpha_0}^{\perp}(\lambda) C_{\varepsilon}^*$. The integral kernels of $N_{\lambda}$ and $R_{\alpha_0}^{\perp}(\lambda)$ are
\beq
\mathcal{N}_{\lambda}(x,u,x',u') := 
\begin{dcases}
\psi_0(u)\frac{\ee^{-\sqrt{\mu_0^2 - \lambda}|x - x'|} - 1}{2 \sqrt{\mu_0^2 - \lambda}}\,\overline{\phi_0(u')} & \mbox{if} \quad n = 1,\\
\frac{1}{2\pi}\psi_0(u)\left(K_0\left(\sqrt{\mu_0^2 - \lambda}|x-x'|\right) + \ln\sqrt{\mu_0^2 - \lambda}\right)\overline{\phi_0(u')} & \mbox{if} \quad n = 2,
\end{dcases}
\eeq
and
\beq \label{projected-kernel}
\mathcal{R}^{\perp}_{\alpha_0}(x,u,x',u';\lambda) := 
\begin{dcases}
\sum_{j=1}^{+\infty}\psi_j(u) \frac{\ee^{-\sqrt{\mu_j^2 - \lambda}|x-x'|}}{2\sqrt{\mu_j^2 - \lambda}}  \,\overline{\phi_j(u')} & \mbox{if} \quad n=1,\\ 
-\frac{1}{2\pi}\sum_{j=1}^{+\infty}\psi_j(u) K_0\left(\sqrt{\mu_j^2 - \lambda}|x-x'|\right)\overline{\phi_j(u')} & \mbox{if} \quad n=2,
\end{dcases}
\eeq
respectively. We see that $R^{\perp}_{\alpha_0}$ is nothing else than the projection of the resolvent of $H_{\alpha_0}$ on higher transversal modes. We define new variable 
\beq
k:= 
\begin{dcases}
\sqrt{\mu_0^2 - \lambda} & \mbox{if} \quad n=1, \\
\left(\ln \sqrt{ \mu_0^2- \lambda}\right)^{-1} & \mbox{if} \quad n=2,
\end{dcases}
\eeq
and show that $M_{\varepsilon}^{\lambda}$ is well-behaved with respect to this variable including the region where $k=0$ (i.e. where $\lambda = \mu_0^2$). This will hold whenever $\beta$ and its derivatives decay sufficiently fast in $\pm \infty$. We divide the proof of this fact into several lemmas.
%
\subsubsection{Behaviour of the projected resolvent}
%
Independently on the specific form of the integral kernel \eqref{projected-kernel} of the projected resolvent $R_{\alpha_0}^{\perp}(\lambda)$, we are able to establish its boundedness and analyticity.
\begin{lem} \label{lemma1}
$D R_{\alpha_0}^{\perp}(\lambda) C_{\varepsilon}^*$ as a function of $k$ defined in $\left\{ k \in \CC \left|\, \Re k > 0\right.\right\}$ for $n=1$ or \\in $\left\{ k \in \CC \left|\, \Re k < 0\right.\right\}$ for $n=2$ is a bounded operator-valued function.
\end{lem}
\begin{proof}
Let us define a projection $\mathcal{P}_0$ onto the subspace in $L^2(\Omega)$ of the functions of the form $\varphi \otimes \psi_0$, where $\varphi \in L^2(\RR^n)$ and $\psi_0$ was defined in \eqref{function}. We denote $\mathcal{P}_0^{\perp} := \JJ - \mathcal{P}_0$ projection onto its orthogonal complement. Now $R_{\alpha_0}^{\perp}(\lambda) = R_{\alpha_0}(\lambda) \mathcal{P}_0^{\perp}$ has an analytic continuation into the region $\CC \setminus [\mu_1^2, +\infty)$ since the lowest point in the spectrum of $H_{\alpha_0}\mathcal{P}_0^{\perp}\upharpoonright \mathcal{P}_0^{\perp}L^2(\Omega)$ is $\mu_{1}^2$. (Recall that its spectrum lies on the positive real half-line.) This includes the studied region $\CC \setminus [\mu_0^2,+\infty)$.
\paragraph{}In fact, we need show that $D R_{\alpha_0}^{\perp}(\lambda) C_{\varepsilon}^*$ is bounded. It is straightforward to see since every action of $C_{\varepsilon}$ on any $\Psi \in L^2(\Omega)$ can be estimated as follows:
\beq \label{odhad-C}
|(C_{\varepsilon}\Psi)_j| \leq
\begin{dcases}
2 \left\|\left|\partial_{x_j}\beta\right|^{1/2}\right\|_{L^{\infty}(\RR^n)}\,d\, |\Psi| \qquad &\mathrm{for}\, j = 1, \dots n,\\
2 \left\||\beta|^{1/2}\right\|_{L^{\infty}(\RR^n)}|\Psi| \qquad &\mathrm{for}\, j = n+1, \\
\left\|\left(\Delta'\beta\right)\right\|_{L^{\infty}(\RR^n)}\,d\,|\Psi| \qquad &\mathrm{for}\, j=n+2\\
\varepsilon \|\partial_{x_j} \beta\|_{L^{\infty}(\RR^n)}\,d^2|\Psi| \qquad &\mathrm{for}\, j = n+3, \dots, 2n+3,
\end{dcases}
\eeq
and we see that $C_{\varepsilon}$ is bounded and the same holds for $C_{\varepsilon}^*$. To show that $D R_{\alpha_0}^{\perp}(\lambda)$ is also bounded, we prepare several estimates. The partial differentiations in $D$ may be estimated by either gradient or identity. We set the constant $c$ equal to the maximum of the norms derived from $\beta$ appearing in \eqref{odhad-C} and estimate
\beq \label{odhad-D}
|(D\P_0^{\perp}\Psi)_j| \leq
\begin{dcases}
c\,|\nabla\P_0^{\perp}\Psi| \qquad & \mathrm{for} \, j= 1, \dots, n+1,\\
c\,|\Psi| &\mathrm{for} \, j= n+2, \dots, 2n+3.
\end{dcases}
\eeq
The action of the gradient on the resolvent may be estimated as well:
\beq \label{dva}
\bal
\|\nabla R_{\alpha_0}^{\perp}(\lambda)\Psi\|_{L^2(\Omega)}^2 
=& \left(\Psi, R_{\alpha_0}^{\perp}(\lambda)\Psi\right) + \lambda \left\|R_{\alpha_0}^{\perp}(\lambda)\Psi\right\|^2\\
&\leq \|R_{\alpha_0}^{\perp}(\lambda)\| \|\Psi\|_{L^2(\Omega)}^2 + |\lambda| \|R_{\alpha_0}^{\perp}(\lambda)\|^2 \|\Psi\|_{L^2(\Omega)}^2.
\eal
\eeq
Putting (\ref{dva}) and (\ref{odhad-D}) together we obtain
\beq
\bal
\|D R^{\perp}_{\alpha_0}(\lambda)\Psi\|_{L^2(\Omega)}^2 
&\leq  (n+1) \|\nabla R^{\perp}_{\alpha_0}(\lambda) \Psi\|_{L^2(\Omega)}^2 + (n+2) \|R^{\perp}_{\alpha_0}(\lambda)\Psi\|_{L^2(\Omega)}^2\\
&\leq (n+1)\|R_{\alpha_0}^{\perp}(\lambda)\| \|\Psi\|_{L^2(\Omega)}^2 + \left((n+1)|\lambda|+n+2\right) \|R_{\alpha_0}^{\perp}(\lambda)\|^2 \|\Psi\|_{L^2(\Omega)}^2\\
&<+\infty.
\eal
\eeq
It follows that $D R^{\perp}_{\alpha_0}(\lambda)$ is a bounded operator on $L^2(\Omega)$ and consequently also $D R^{\perp}_{\alpha_0}(\lambda)C_{\varepsilon}^*$.
\end{proof}
\begin{lem} \label{lemma2}
$(\psi, D R_{\alpha_0}^{\perp}(\lambda) C_{\varepsilon}^* \phi)$ as a function of $k$ is analytic in $\left\{ k \in \CC \left|\, \Re k > 0\right.\right\}$ for $n=1$ or in $\left\{ k \in \CC \left|\, \Re k < 0\right.\right\}$ for $n=2$ for every $\psi, \phi \in L^2(\Omega) \otimes \CC^{2n+3}$.
\end{lem}
\begin{proof}
The analyticity can be showed in the same manner as the boundedness in Lemma \ref{lemma1}, now using the first resolvent formula. It is equivalent to showing that the sesquilinear form 
\beq
r_{\lambda}(\Phi,\Psi) := \left(\Phi, D R_{\alpha_0}^{\perp}(\lambda) C_{\varepsilon}^* \Psi\right) 
\eeq
 is analytic as a function of $\lambda$ for every $\Phi$ and  $\Psi$ from the fundamental subset. We are in fact able to show for every $\Phi, \Psi \in L^2(\Omega)$ and every $\lambda_0 \in \CC\setminus[\mu_0^2,+\infty)$
\beq
\bal
r_{\lambda_0}'(\Phi,\Psi) &:= \lim_{\lambda\rightarrow\lambda_0} \frac{r_{\lambda}(\Phi,\Psi)-r_{\lambda_0}(\Phi,\Psi)}{\lambda - \lambda_0} \\
&=  \lim_{\lambda\rightarrow\lambda_0}\frac{\left(\Phi, \left(D R_{\alpha_0}^{\perp}(\lambda) C_{\varepsilon}^* - D R_{\alpha_0}^{\perp}(\lambda_0) C_{\varepsilon}^*\right)\Psi\right)}{\lambda-\lambda_0}\\
&=  \lim_{\lambda\rightarrow\lambda_0}\frac{\left(\Phi, D \left((\lambda - \lambda_0)R_{\alpha_0}^{\perp}(\lambda)R_{\alpha_0}^{\perp}(\lambda_0)\right) C_{\varepsilon}^* \Psi \right)}{\lambda-\lambda_0}\\
&=  \lim_{\lambda\rightarrow\lambda_0}\frac{\left(D^* \Phi, \left((\lambda - \lambda_0)R_{\alpha_0}^{\perp}(\lambda)R_{\alpha_0}^{\perp}(\lambda_0)\right) C_{\varepsilon}^* \Psi \right)}{\lambda-\lambda_0}\\
&= \left(\Phi, D R_{\alpha_0}^{\perp}(\lambda_0)^2 C_{\varepsilon}^* \Psi \right).
\eal
\eeq
(The dash denotes differentiation with respect to $\lambda$.) The next step would be to show boundedness of $D R_{\alpha_0}^{\perp}(\lambda_0)^2 C_{\varepsilon}^*$ which can be done exactly in the same way as the proof of the boundedness of $D R_{\alpha_0}^{\perp}(\lambda_0) C_{\varepsilon}^*$.
\end{proof}
%
\subsubsection{Behaviour of \texorpdfstring{$N_{\varepsilon}^{\lambda}$}{N epsilon lambda} in the strip (n=1)}
%
Let us now assume decay of $\beta$ and of its derivatives in $\pm \infty$, specifically
\beq \label{predpoklady}
\bal
\lim_{|x| \rightarrow +\infty} |x|^{5+\delta}\,\beta(x) &= 0, \\
\lim_{|x| \rightarrow +\infty} |x|^{5+\delta}\,\partial_{x_j}\beta(x) &= 0,\\
\lim_{|x| \rightarrow +\infty} |x|^{5+\delta}\,\partial_{x_j}^2\beta(x) &= 0,
\eal
\eeq
for all $j= 1, \dots, n$ and any $\delta > 0$. Then we are able to show that $D N_{\lambda} C_{\varepsilon}^*$ is well-behaved.
\begin{lem} \label{lemma3}
Let us assume \eqref{predpoklady}. Then $D N_{\lambda} C_{\varepsilon}^*$ as a function of $k$ defined in $\left\{ k \in \CC \left|\, \Re k > 0\right.\right\}$ is a bounded and analytic operator-valued function.
\end{lem}
\begin{proof}
We are able to obtain an integral operator from $D N_{\lambda} C_{\varepsilon}^*$ by immersing the differentiations in $D$ into the inside of the integral operator $N_{\lambda} C_{\varepsilon}^*$. (This operation is justified, if the new integral kernel will be integrable and that is the object of our proof anyway.) Now, in the integral kernel, every part depending on $u$ can be uniformly estimated. Therefore we may check only the boundedness and analyticity of integral operators $h \tilde{N_{\lambda}}h$ and $h \partial\tilde{N_{\lambda}}h$ with kernels $h n_{\lambda} h$ and $h \partial n_{\lambda} h$, respectively, where
\beq
\bal
n_{\lambda}(x,x') &:= \frac{\ee^{-\sqrt{\mu_0^2 - \lambda}|x - x'|} - 1}{2 \sqrt{\mu_0^2 - \lambda}},\\
\partial n_{\lambda}(x,x') &:= -\frac{1}{2} \frac{x - x'}{|x - x'|} \ee^{-\sqrt{\mu_0^2 - \lambda}|x-x'|},
\eal
\eeq
with $h(x)$ being a bounded continuous function in $\RR$. Its specific form is not important, the main role plays its behaviour in infinity. Since $h$ arises from the terms inside of $C_{\varepsilon}$ and $D$, $h$ decays in $\pm \infty$ faster than $|x|^{5/2 + \delta/2}$. As a consequence, $h \in L^2(\RR,(1 +x^2+x^4) \dx)$ since it is bounded and its absolute value can be estimated near $\pm \infty$ by $1/|x|^{5/2 +\delta/2}$. Using the Hilbert-Schmidt norm we get
\beq
\bal \label{odhad-vyzkumak}
\left\|h \tilde{N_{\lambda}}h\right\|^2 
&\leq \frac{1}{4}\int\limits_{\mathbb{R}^2}|h(x)|^2 \,|x - x'|^2\,|h(x')|^2 \dx \dx' \\
&\leq \frac{1}{2}\int\limits_{\mathbb{R}^2}|h(x)|^2 \,\left(|x|^2 + |x'|^2\right)\,|h(x')|^2 \dx \dx' \\
&\leq \frac{1}{2}\left(\int \limits_{\mathbb{R}}|h(x)|^2(1 + x^2) \dx \right)^2 \!< +\infty.
\eal
\eeq
In the same manner the boundedness of $h \partial\tilde{N_{\lambda}}h$ can be shown:
\beq
\bal
\left\|h \tilde{N_{\lambda}}'h\right\| 
&\leq \frac{1}{4} \int_{\RR^2} |h(x)^2| \ee^{-2\sqrt{\mu_0^2 - \lambda}|x-x'|}|h(x')|^2 \dx \dx' \\
&\leq \frac{1}{4} \int_{\RR}|h(x)^2|\dx \int_{\RR}|h(x')^2|\dx' < + \infty.
\eal
\eeq
To verify the second inequality in (\ref{odhad-vyzkumak}) it is sufficient to see that
\beq \label{trapny-odhad}
\left|\frac{\ee^{a + \ii b} - 1}{-(a + \ii b)}\right|^2 \leq 1
\eeq
holds for all $a,b \in \mathbb{R}, a < 0$. After an explicit calculation of the absolute value on left-hand side of the inequality and a simple algebraic manipulation, we reformulate our problem to verification that
\beq
1 + \ee^{2a} - 2 \ee^a \cos b - a^2 - b^2 \leq 0 
\eeq
holds. We employ the estimate $\cos b \geq 1 - b^2/2$ which holds for all $b \in \mathbb{R}$ to get
\beq
\bal
1 + \ee^{2a} - 2 \ee^a \cos b - a^2 - b^2 &\leq 1 + \ee^{2a} - 2 \ee^a\left(1 - \frac{b^2}{2}\right) - a^2 - b^2 \\
&\leq 1 + \ee^{2a} - 2 \ee^a1 + b^2 - a^2 - b^2 \\
&= 1 + \ee^{2a} - 2 \ee^a - a^2.
\eal
\eeq
Using calculus of functions of one variable it is now easy to check that $f(a) := 1 + \ee^{2a} - 2 \ee^a - a^2 \leq 0$.
\paragraph{}For proving the analyticity we need to check the finiteness of the norms of derivatives of the integral kernels
\beq
\bal
\frac{\mathrm{d}n_{\lambda}}{\mathrm{d}k}(x,x') &= \frac{-k|x - x'|\ee^{- k|x - x'|} - \ee^{-k|x - x'|}+1}{2 k^2} \\
\frac{\mathrm{d}n_{\lambda}'}{\mathrm{d}k}(x,x') &=\frac{1}{2} (x - x')  \ee^{-k |x-x'|}.
\eal
\eeq 
We estimate
\beq
\left|\frac{-k|x - x'|\ee^{- k|x - x'|} - \ee^{-k|x - x'|}+1}{2 k^2}  \right|\leq |x-x'|^2.
\eeq
(This can be proven in exactly the same way as \eqref{trapny-odhad}). Similarly as in \eqref{odhad-vyzkumak} we calculate the bound and we obtain
\beq
\bal
\left\|h \frac{\mathrm{d}n_{\lambda}'}{\mathrm{d}k} h\right\| 
&\leq \int\limits_{\mathbb{R}^2}|h(x)|^2 |x-x'|^4|h(x')|^2\dx \dx' \\
&\leq 8 \int\limits_{\mathbb{R}^2} |h(x)|^2 (|x|^4 + |x'|^4)|h(x')|^2\dx \dx' \\
&\leq 8\left( |h(x)|^2 (1 + |x|^4)\dx\right)^2 < + \infty.
\eal
\eeq
We conduct the estimate of $\frac{\mathrm{d}n_{\lambda}'}{\mathrm{d}k}$ in the same way:
\beq
\bal
\left\|h \frac{\mathrm{d}n_{\lambda}}{\mathrm{d}k} h\right\| 
&\leq\int\limits_{\mathbb{R}^2}|h(x)|^2 \left(|x|^2 + |x'|^2\right)|h(x')|^2\dx \dx' \\
&\leq \left( |h(x)|^2 (1 + |x|^2)\dx\right)^2 < + \infty.
\eal
\eeq

\end{proof}
%
\subsubsection{Behaviour of \texorpdfstring{$N_{\varepsilon}^{\lambda}$}{N epsilon lambda} in the layer (n=2)}
%
For the layer, there is a different requirement on the decay of $\beta$ and of its derivatives in $\pm \infty$, specifically
\beq \label{predpoklady2}
\bal
\lim_{|x| \rightarrow +\infty} |x|^{4 + \delta}\,\beta(x) &= 0, \\
\lim_{|x| \rightarrow +\infty} |x|^{4 + \delta}\,\partial_{x_j}\beta(x) &= 0,\\
\lim_{|x| \rightarrow +\infty} |x|^{4 + \delta}\,\partial_{x_j}^2\beta(x) &= 0,
\eal
\eeq
for all $j= 1, \dots, n$, where $\delta$ is an arbitrarily small positive number. Note that these conditions differ from \eqref{predpoklady}. This is caused by both different dimension of the problem and by using a different estimate method.
\begin{lem} \label{lemma5}
Let us assume \eqref{predpoklady2}. Then $D N_{\lambda} C_{\varepsilon}^*$ as a function of $k$ defined in $\left\{ k \in \CC \left|\, \Re k < 0\right.\right\}$ is a bounded and analytic operator-valued function
\end{lem}
\begin{proof}
Throughout this proof we employ various properties of the Macdonald function $K$ which can be found e.g. in \cite[9.6-7]{Abramowitz}. Similarly as in the proof of Lemma \ref{lemma3} we get rid of the derivatives in $D$ and we may check the boundedness of integral operators $h \tilde N_{\lambda}h$ and $h \partial_{\mu}\tilde N_{\lambda} h$ with kernels $h n_{\lambda} h$ and $h \partial_{\mu}n_{\lambda} h$, respectively, where
\beq
\bal
n_{\lambda}(x,x') &:= \frac{1}{2\pi}\left(K_0\left(w_0(\lambda)|x-x'|\right) + \ln w_0(\lambda)\right),\\
\partial_{\mu} n_{\lambda}(x,x') &:= -\frac{1}{2\pi} \frac{x_{\mu} - x'_{\mu}}{|x - x'|}  w_0(\lambda) K_1(w_0(\lambda) |x-x'|),
\eal
\eeq
with $\mu = 1,2$ and x $\partial_{\mu}$ means the derivative with respect to $x_{\mu}$. We adopted the notation $w_0(\lambda) = \sqrt{\mu_0^2 - \lambda} $. We used the differentiation formula for Macdonald functions, $K_0' = - K_1$. For the purpose of the estimates, we use several other formulae, which are valid for any $z \in (0, +\infty)$:
\beq
\bal
\left| (K_0(z) + \ln z)\ee^{-z}\right| &\leq c_1,\\
\left| K_1(z) - 1/z\right| &\leq c_2,\\
\left| K_1(z) - z\left(K_0(z) + K_2(z)\right)/2\right| &\leq c_3,\\
\left| z K_1(z)\right| &\leq 1,\\
\left| \left(K_0(z) + \ln z\right)/z\right| & \leq c_4.
\eal
\eeq
In the calculation of the integral bounds we make use of the polar coordinates
\beq
(x'_1,x_2') = (x_1 - \rho \cos \varphi, x_2 - \rho \sin \varphi)
\eeq
and employ the estimate via Schur-Holmgren bound, holding for every integral operator $K$ with the integral kernel $\mathscr{K}(\cdot,\cdot)$ acting on $L^2(M)$, where $M$ is an open subset of $\RR^n$\cite[Lem. 2.2]{BEGK}:
\beq \label{Schur-Holmgren}
\|K\| \leq \|K\|_{SH} := \left( \sup_{x\in M} \int_M |\mathscr{K}(x,y)|\dy \;\sup_{y\in M} \int_M |\mathscr{K}(x,y)|.\dx \right)^{1/2}
\eeq
Since $h$ is continuous, bounded and $|x||h(x)| \leq 1/|x|^{1+\delta}$ for sufficiently high $|x|$, then $h \in L^1(\RR^2, (1+|x|)\dx)$. We obtain
\beq
\bal
\left\|h \tilde{N_{\lambda}} h\right\| \leq& \frac{1}{2\pi} \sup_{x\in\RR^2} |h(x)| \int_{\RR^2}\left|\left(K_0\left(w_0(\lambda)|x-x'|\right) + \ln w_0(\lambda)|x-x'| - \ln|x-x'|\right) h(x') \right|\dx' \\
\leq& c_1 \|h\|^2_{L^{\infty}(\RR^2)} \left(\int_0^{R}\ee^{w_0(\lambda)\rho}\rho \drho + \int_{0}^{R}|\ln \rho|\rho\drho\right) \\
&+ \frac{1}{2 \pi} \sup_{x\in\RR^2} |h(x)|\sup_{z\in (R,+\infty)} \frac{K_0(w_0(\lambda)z) - \ln w_0(\lambda) z + \ln z}{z}\int_{\RR^2}(|x| + |x'|)h(x')\dx'\\
\leq& c_1 \|h\|^2_{L^{\infty}(\RR^2)} R \left( R \ee^{w_0(\lambda)R} + \max\left\{ \ee^{-1}, R \ln R\right\}\right) \\
& + (c_4 + c_5) \left(\sup_{x\in\RR^2} |x h(x)| \|h\|_{L^1(\RR^2)} + \sup_{x\in\RR^2} |h(x)| \| h\|_{L^1(\RR^2,|x|\dx)}\right) < + \infty,
\eal
\eeq
where $R > 0$ arbitrary and $c_5 := \sup_{z\in (R,+\infty)} \ln z / z$. The estimates of $\|h \tilde \partial_{\mu} N_{\lambda} h\|_{SH}$ yield
\beq
\bal
\left\|h \partial_{\mu}\tilde N_{\lambda} h\right\|_{SH} 
 \leq& \|h\|^2_{L^{\infty}(\RR^2)} \int_0^R \frac{\rho \drho}{\rho} + \sup_{x\in \RR^2}|h(x)|w_0(\lambda)\sup_{z\in(R,+\infty)} K_1(w_0(\lambda) z) \|h\|_{L^1(\RR^2)} \\
 \leq & \|h\|^2_{L^{\infty}(\RR^2)} R + \frac{1}{R} \|h\|_{L^{\infty}(\RR^2)} \|h\|_{L^1(\RR^2)} < +\infty.
\eal
\eeq
\paragraph{}Checking the analyticity means, according to its definition, checking the analyticity of the two sesquilinear forms $(\Phi, N_{\lambda}\Psi)$ and $(\Phi, \partial_{\mu}N_{\lambda}\Psi)$ with arbitrary $\Phi,\Psi \in L^2(\RR^2)$, taken as functions of $k$. This can be done by checking the finiteness of the norms of $\text{d}N_{\lambda}/\text{d}k$ and $\text{d}(\partial_{\mu}N_{\lambda})/\text{d}k$. Using the formula $K_1'(z) = (K_0(z) + K_2(z))/2$ and employing the notation $z:= w_0(\lambda)|x-x'|$ we arrive at 
\beq
\bal
\frac{\text{d} n_{\lambda}}{\text{d}k} &= \frac{1}{2 \pi} \frac{z}{k^2} \left(K_1(z) - \frac{1}{z} \right), \\
\frac{\text{d} (\partial_{\mu} n_{\lambda})}{\text{d}k} &= \frac{1}{2\pi}\frac{x_{\mu} - x'_{\mu}}{|x - x'|} \frac{w_0(\lambda)}{k^2}\left( K_1(z) - z \frac{K_0(z) + K_2(z)}{2}\right).
\eal
\eeq
Now we use the inequality $\ee^{k^{-1}}/k^2 \leq c_6$, valid for all $k \in (-\infty,0)$ and estimate
\beq
\bal
\left\|h\frac{\text{d} n_{\lambda}}{\text{d}k}h\right\| \leq& \frac{c_2 c_6}{2\pi} \sup_{x \in \RR^2} |h(x)| \int_{\RR^2} \left(|x| + |x'|\right) |h(x')| \dx'\\
\leq& \frac{c_2 c_6}{2\pi}\left( \sup_{x\in\RR^2} |x h(x)| \|h\|_{L^1(\RR^2)} + \sup_{x\in\RR^2} |h(x)| \|h\|_{L^1(\RR^2,x\dx)} \right) < +\infty.
\eal
\eeq
The estimate of $\text{d} \partial_{\mu}n_{\lambda}/\text{d}k$ can also be carried out without further difficulties:
\beq
\bal
\left\|h \frac{\text{d} (\partial_{\mu} n_{\lambda})}{\text{d}k} h\right\| 
&\leq \frac{c_3 c_6}{2\pi} \|h\|_{L^{\infty}(\RR^2)} \|h\|_{L^1(\RR^2)} <+\infty.
\eal
\eeq
\end{proof}
%
\subsection{The bound state}
%
Now we are able to summarise the results about both parts of $M_{\varepsilon}^{\lambda}$ and state that it is well-behaved in the right half-plane, as we suspected.
\begin{lem} \label{WClemma}
Let us assume \eqref{predpoklady} if $n=1$ or \eqref{predpoklady2} $n=2$. Then $M_{\varepsilon}^{\lambda}(\lambda(k))$ as a function of $k$ defined in $\left\{ k \in \CC \left|\, \Re k > 0\right.\right\}$ for $n=1$ or in $\left\{ k \in \CC \left|\, \Re k < 0\right.\right\}$ for $n=2$ is a bounded and analytic operator-valued function which can be analytically continued to the region $\left\{ k \in \CC \left|\, \Re k \geq 0\right.\right\}$ or $\left\{ k \in \CC \left|\, \Re k \leq 0\right.\right\}$, respectively.
\end{lem} 
\begin{proof}
Using Lemmas \ref{lemma1}, \ref{lemma2} and Lemma \ref{lemma3} in the case of the strip or Lemma \ref{lemma5} in the case of the layer, we see that $M_{\varepsilon}^{\lambda}(\lambda(k))$ and its derivatives are bounded when $\Re k \rightarrow 0$, therefore  $M_{\varepsilon}^{\lambda}(\lambda(k))$ can be analytically continued to the region where $\Re k=0$.
\end{proof}
Equipped with Lemma \ref{WClemma} we may proceed to the main proof of this section.
%
\subsubsection{Proof of Theorem \ref{vetavet}}
%
Our goal is to find the condition to ensure that the operator $\varepsilon K_{\varepsilon}^{\lambda}$ has an eigenvalue $-1$. First we restrict ourselves to the case $n=1$. Using Lemma \ref{WClemma} we may choose $\varepsilon$ so small that $\|M_{\varepsilon}^{\lambda}\| < 1$ so the operator $(I + \varepsilon M_{\lambda})^{-1}$ exists and is analytic in the region $\left\{ k \in \CC \left|\, \Re k \geq 0\right.\right\}$. We may write
\beq
\bal
(I + K_{\varepsilon}^{\lambda})^{-1} &= \left( (I + M_{\varepsilon}^{\lambda})(I + (I+ M_{\varepsilon}^{\lambda})^{-1} L_{\varepsilon}^{\lambda})\right)^{-1} = \left(I + (I + M_{\varepsilon}^{\lambda})^{-1}L_{\varepsilon}^{\lambda} \right)^{-1}(I + M_{\varepsilon}^{\lambda})^{-1}.
\eal
\eeq
and therefore only determine whether the operator $P^{\lambda}_{\varepsilon}:=(I+ M_{\varepsilon}^\lambda)^{-1} L_{\varepsilon}^{\lambda}$ has eigenvalue $-1$. Since $L_{\varepsilon}^{\lambda}$ is a rank-one operator by definition, we can write
\beq
P^{\lambda}_{\varepsilon}(\cdot) = \Phi (\Psi, \cdot),
\eeq
with
\beq
\bal
\overline{\Psi(x,u)}&:= \varepsilon \psi_0(u)\frac{1}{2 \sqrt{-\lambda}}C^*_{\varepsilon}, \\
\Phi(x,u)&:=\left((I +  M_{\varepsilon}^{\lambda})^{-1}D \overline{\phi_0}\right)(x,u).
\eal
\eeq
(Recall that $C_{\varepsilon}^*$ is just an operator of multiplication by a function.) The operator $P^{\lambda}_{\varepsilon}$ can have only one eigenvalue, namely $(\Psi,\Phi)$. Putting it equal to $-1$ we get the condition
\beq \label{implicit}
-1 = \frac{\varepsilon}{2 \sqrt{\mu_0^2 -\lambda}} \int_{\Omega} \psi_0(u) \left(C_{\varepsilon}^*(I +  M_{\varepsilon}^{\lambda})^{-1}D \overline{\phi_0}\right)(x,u) \dx \du.
\eeq
Let us define the function
\beq
G(k,\varepsilon) := - \frac{\varepsilon}{2} \int_{\Omega} \psi_0(u) \left(C_{\varepsilon}^*(I +  M_{\varepsilon}^{\lambda})^{-1}D \overline{\phi_0}\right)(x,u) \dx \du.
\eeq
We shall return to the proof of existence of the eigenvalue later on, let us now for a moment assume that there is a solution to the implicit equation \eqref{implicit}. Using the formula
\beq
(I +  M_{\varepsilon}^{\lambda})^{-1} = I - M_{\varepsilon}^{\lambda}(I + M_{\varepsilon}^{\lambda})^{-1}  = I - M_{\varepsilon}^{\lambda} +  \left(M_{\varepsilon}^{\lambda}\right)^2 (I + M_{\varepsilon}^{\lambda})^{-1}
\eeq
we derive its asymptotic expansion in $0$:
\beq
\bal
k(\varepsilon) &=  \frac{\varepsilon}{2} \int_{\Omega} \psi_0 C_{0}^*D \overline{\phi_0} + \mathcal{O}(\varepsilon^2)= \frac{\varepsilon}{2}\left(C_{\varepsilon}^*D \phi_0,\psi_0\right) + \mathcal{O}(\varepsilon^2)
\eal
\eeq
for $\varepsilon$ tending to $0$. Since $B_j \phi_0 = 0$ for $j=1,\dots,n$, $\int_{\RR^n} \Delta'\beta(x) \dx = 0$ (due to the decay in infinity) and $(C_{\varepsilon})_{l} = \mathcal{O}(\varepsilon)$ for $l = n+3,\dots,2n+3$, after simple calculation we have
\beq \label{neco}
\bal
k(\varepsilon)&=\frac{\varepsilon}{2} \left(B_{n+1} \phi_0, A_{n+1} \psi_0\right)+ \mathcal{O}(\varepsilon^2) = \ii \varepsilon \langle\beta\rangle \left(\frac{\partial}{\partial u} \phi_0 ,\psi_0\right) +\mathcal{O}(\varepsilon^2) = - \varepsilon\langle\beta\rangle \alpha_0 + \mathcal{O}(\varepsilon^2).
\eal
\eeq
Here we used $\alpha_0 < \pi/d$. Clearly $k \rightarrow 0$ when $\varepsilon \rightarrow 0$ and if $\lambda$ ought to be an eigenvalue outside the essential spectrum, $\Re k \geq 0$ must hold. This is if $\langle\beta\rangle \alpha_0 < 0$. If $\langle\beta\rangle \alpha_0 > 0$ no eigenvalue can exist. The expansion of $k$ reads $k(\varepsilon) = \sqrt{\mu_0^2 - \lambda} = \varepsilon \langle\beta\rangle \alpha_0 + \mathcal{O}(\varepsilon^2)$ and this gives
\beq
\lambda(\varepsilon) = \mu_0^2 - \varepsilon^2 \langle\beta\rangle^2 \alpha_0^2 + \mathcal{O}(\varepsilon^3)
\eeq
as $\varepsilon$ goes to $0$.
\paragraph{}So far we only found out what our solution had to meet, if it existed. Equipped with the knowledge of the asymptotic expansion \eqref{neco} we apply the Rouch\'{e}'s theorem \cite[Thm. 10.43 b)]{Rudin87} in the disc $B(k_0,r)$, where 
\beq
k_0 := - \varepsilon\langle\beta\rangle \alpha_0
\eeq
and the radius $r$ is so small that the whole disc lies in the half-plane $\Re k > 0$. First we show that $G(k,\varepsilon)$ is analytic as a function of $k$ in the region $\Re k \geq 0$. We prepare formula for differentiating of $\left( 1 + M_{\varepsilon}^{\lambda}\right)^{-1}$:
\beq
\bal
\frac{\partial}{\partial k} \left( 1 + M_{\varepsilon}^{\lambda}\right)^{-1} &= \lim_{k' \rightarrow k} \frac{\left( 1 + M_{\varepsilon}^{\lambda}\right)^{-1} - \left( 1 + M_{\varepsilon}^{\lambda'}\right)^{-1}}{ k - k'} \\
&= \lim_{k' \rightarrow k} \frac{\left( 1 +  M_{\varepsilon}^{\lambda}\right)^{-1}( M_{\varepsilon}^{\lambda} - \varepsilon M_{\varepsilon}^{\lambda'})\left( 1 + M_{\varepsilon}^{\lambda'}\right)^{-1}}{ k - k'}\\
&= \left( 1 + M_{\varepsilon}^{\lambda}\right)^{-1}\frac{\partial M_{\varepsilon}^{\lambda}}{\partial k}  \left( 1 +  M_{\varepsilon}^{\lambda}\right)^{-1}.
\eal
\eeq
And we have for $G(k,\varepsilon)$ in the region $\Re k \geq 0$:
\beq \label{odhadG}
\bal
\left|\frac{\partial G(k,\varepsilon)}{\partial k}\right| &= \frac{\varepsilon}{2}\left|\int_{\Omega} \psi_0(u)\left(C_{\varepsilon}^* \frac{\partial}{\partial k} \left( 1 + M_{\varepsilon}^{\lambda}\right)^{-1} D \overline{\phi_0} \dx \du \right)(x,u)\right|\\
&= \frac{\varepsilon}{2}\left|\int_{\Omega} \psi_0(u)\left(C_{\varepsilon}^*\left( 1 + M_{\varepsilon}^{\lambda}\right)^{-1}\frac{\partial M_{\varepsilon}^{\lambda}}{\partial k}  \left( 1 +  M_{\varepsilon}^{\lambda}\right)^{-1}D \overline{\phi_0} \dx \du \right)(x,u)\right|\\
&\leq \frac{\varepsilon}{2} \|\psi_0\|_{L^2(I)} \|C_{\varepsilon}^*\| \left\|\left( 1 +  M_{\varepsilon}^{\lambda}\right)^{-1}\right\|^2 \left\|\frac{\partial M_{\varepsilon}^{\lambda}}{\partial k}\right\| \|D \phi_0\|_{L^2(\Omega)} \\
& = K \varepsilon,
\eal
\eeq
where we used analyticity of $M_{\varepsilon}^{\lambda}$ in the region $\Re k \geq 0$ (Lemma \ref{WClemma}) and properties of operators $C_{\varepsilon}^*$ and $D$. With sufficiently small $r$ we can expand $G(k,\varepsilon)$ in Taylor series in the neighbourhood of the point $k_0$
\beq \label{Taylor}
G(k,\varepsilon) = G(k_0,\varepsilon) + (k - k_0) \frac{\partial G(k, \varepsilon)}{\partial k}(k_0) + \mathcal{O}((k- k_0)^2),
\eeq
We employ Rouch\'{e}'s theorem to show that the equation \eqref{implicit} possesses one simple and unique solution in the half-plane $\textrm{Re}\, k >0$. We prove that the holomorphic functions $G(k, \varepsilon) - k$ and $k_0 - k$ have the same number of zeros (counted as many times as their multiplicity) in $B(k_0,r)$ (i.e. one simple zero). It suffices to show that absolute value of their difference, $\left|G(k,\varepsilon) - k_0\right|$, is strictly smaller than $|k_0 - k|$. It directly follows for all $k \in B(k_0,r)$ from \eqref{neco}, \eqref{odhadG} and \eqref{Taylor}
\beq
\left|G(k,\varepsilon) - k_0 \right| \leq \left|\frac{\partial G(k, \varepsilon)}{\partial k} + o(1)\right| |k - k_0|, 
\eeq
where $o(1)$ tends to $0$ as $k$ tends to $k_0$. Using \eqref{odhadG} and setting $\varepsilon$ and $r$ sufficiently small, we can make the coefficient by $|k - k_0|$ strictly smaller than $1$.
\paragraph{}The reality of the obtained eigenvalue is ensured by the $\PT$-symmetry of the operator $H_{\alpha}$ (cf. Proposition \ref{PTH}). Indeed, from the relation \eqref{PTsymmetry} follows that if $\lambda$ is an eigenvalue of $H_{\alpha}$, then $\overline{\lambda}$ is its eigenvalue as well. From the uniqueness follows that $\lambda=\overline{\lambda}$ and it is therefore real.
\paragraph{}The proof for the case $n=2$ proceeds in the same manner. The equation \eqref{implicit} becomes
\beq
-1 = -\frac{\varepsilon}{2 \pi}\ln \sqrt{\mu_0^2-\lambda} \int_{\Omega} \psi_0(u) \left(C_{\varepsilon}^*(I +  M_{\varepsilon}^{\lambda})^{-1}D \overline{\phi_0}\right)(x,u) \dx \du
\eeq
and solving it yields the asymptotic expansion
\beq
\bal
k(\varepsilon) &= -\frac{\varepsilon}{2\pi}\left(C_{\varepsilon}^* D \phi_0,\psi_0\right) + \mathcal{O}(\varepsilon^2)\\
&= \frac{\varepsilon}{2 \pi} \langle \beta\rangle \alpha_0 + \mathcal{O}(\varepsilon^2).
\eal
\eeq
Now from the requirement that $\Re k \leq 0$ must hold, we obtain the condition $\langle \beta\rangle \alpha_0 > 0$ again. The expansion of $\lambda(\varepsilon)$ reads
\beq
\lambda(\varepsilon) = \mu_0^2 - \ee^{2/w(\varepsilon)} + \mathcal{O}(\varepsilon^3),
\eeq
where $w(\varepsilon) = \frac{\varepsilon}{\pi} \langle\beta\rangle \alpha_0$, for $\varepsilon \rightarrow 0$. The proof of existence and uniqueness holds without change. 
\begin{rmrk} \label{singularita-rmrk}
Note the important role of the singularity of the resolvent function on the existence of the bound state. For this purpose it was necessary for $K_{\varepsilon}^{\lambda}$ to have an eigenvalue $-1$, a necessity for this is $\|K_{\varepsilon}^{\lambda}\| \geq 1$. It would not be possible in the limit $\varepsilon \rightarrow 0$ if the resolvent function inside $K_{\varepsilon}^{\lambda}$ had not a singularity in the limit $\lambda \rightarrow \mu_0^2$. Since the resolvent function in dimension $n\geq 3$ does not possess a singularity, it can not be expected that a weak perturbation of the boundary would yield a bound state. More likely there would be a critical value of the parameter $\varepsilon$, giving a lower bound on $\varepsilon$ enabling a bound state.
\end{rmrk}
\section*{Acknowledgements}
The research was supported by the Czech Science Foundation within the project 14-06818S and by Grant Agency of the Czech Technical University in Prague, grant No. SGS13/217/OHK4/3T/14. The author would like to express his gratitude to David Krej\v ci\v r\' ik and Petr Siegl for valuable discussions and comments.
%
\bibliographystyle{acm}
%

%
\end{document}